\documentclass[11pt, reqno]{amsart}
\usepackage[dvipsnames]{xcolor}
\usepackage{amscd, amsmath, amssymb, amsthm, yfonts, mathrsfs, hhline, tikz}
\usetikzlibrary{arrows, decorations.markings, patterns}
\tikzset{>=latex}
\usepackage[centering,margin=1.2in]{geometry}
\usepackage[matrix,arrow,curve]{xy}

\global\binoppenalty=10000

\newtheorem{theorem}{Theorem}[section]
\newtheorem{lemma}[theorem]{Lemma}
\newtheorem{prop}[theorem]{Proposition}
\newtheorem{cor}[theorem]{Corollary}

\theoremstyle{definition}
\newtheorem{defn}[theorem]{Definition}
\newtheorem{remark}[theorem]{Remark}
\newtheorem{example}[theorem]{Example}
\newtheorem{notation}[theorem]{Notation}

\numberwithin{equation}{section}

\def\beq{\begin{equation}}
\def\eeq{\end{equation}}

\def\ep{\epsilon}

\newcommand{\hr}[1]{\left(#1\right)}                                                    
\newcommand{\hm}[1]{\left|#1\right|}                                                    
\newcommand{\hn}[1]{\left\|#1\right\|}                                                  
\newcommand{\ha}[1]{\left\langle#1\right\rangle}                                        
\newcommand{\hc}[1]{\left\{#1\right\}}                                                  

\def\le{\leqslant}
\def\ge{\geqslant}

\def\Ad{\operatorname{Ad}}

\def\bs{\boldsymbol}

\def\C{\mathbb C}
\def\Cc{\mathcal C}

\def\eps{\varepsilon}

\def\Fc{\mathcal F}
\def\Frac{\operatorname{Frac}}

\def\gl{\mathfrak{gl}}

\def\Hc{\mathcal H}

\def\inv{\mathrm{inv}}
\def\Ic{\mathcal I}
\def\Id{\operatorname{Id}}

\def\la{\lambda}
\def\La{\Lambda}

\def\mf{\mathfrak m}

\def\PV{\operatorname{PV}}

\def\Qc{\mathcal Q}

\def\R{\mathbb R}

\def\Wc{\mathcal W}

\def\Z{\mathbb Z}

\begin{document}

\title{On $b$\,-Whittaker functions}
\author{Gus Schrader}
\author{Alexander Shapiro}

\dedicatory{To Kolya Reshetikhin with admiration}

\begin{abstract}
The $b$\,-Whittaker functions are eigenfunctions of the modular $q$\,-deformed $\gl_n$ open Toda system introduced by Kharchev, Lebedev, and Semenov-Tian-Shansky. Using the quantum inverse scattering method, the named authors obtained a Mellin--Barnes integral representation for these eigenfunctions. In the present paper, we develop the analytic theory of the $b$\,-Whittaker functions from the perspective of quantum cluster algebras. We obtain a formula for the modular open Toda system's Baxter operator as a sequence of quantum cluster transformations, and thereby derive a new modular $b$\,-analog of Givental's integral formula for the undeformed Whittaker function. We also show that the $b$\,-Whittaker functions are eigenvectors of the Dehn twist operator from quantum higher Teichm\"uller theory, and obtain $b$\,-analogs of various integral identities satisfied by the undeformed Whittaker functions, including the continuous Cauchy--Littlewood identity of Stade and Corwin--O'Connell--Sepp\"al\"ainen--Zygouras. Using these results, we prove the unitarity of the $b$\,-Whittaker transform, thereby completing the analytic part of the proof of the conjecture of Frenkel and Ip on tensor products of positive representations of $U_q(\mathfrak{sl}_n)$, as well as the main step in the modular functor conjecture of Fock and Goncharov. We conclude by explaining how the theory of $b$\,-Whittaker functions can be used to derive certain hyperbolic hypergeometric integral evaluations found by Rains.

\end{abstract}

\maketitle

\section{Introduction}

The following application of representation theory to quantum integrability was discovered in the 1970's by Kostant \cite{Kos79}.  Let $G$ be a simply-connected semisimple complex Lie group with positive and negative maximal unipotent subgroups $N_\pm$ and a maximal torus~$H$. If we fix two holomorphic non-degenerate characters $\chi_{\pm} \colon N_\pm\to\C$, a {\em~Whittaker function} with characters $\chi_\pm$ is a holomorphic function $\psi$ on the big cell $G_0=N_-HN_+$ satisfying $\psi(n_-an_+) =  \chi_-(n_-)\psi(a)\chi(n_+)$, for all $n_\pm\in N_\pm$ and $a\in H$.  Kostant observed that the restriction of the Laplacian of $G$ to the space of Whittaker functions is the Hamiltonian of the quantum Toda system, with the higher integrals of motion being delivered by the higher Casimirs of $G$. Thus, the representation theory of semisimple Lie groups controls the spectral theory of the quantum Toda system, the latter subject having been given definitive treatment by Semenov-Tian-Shansky in~\cite{Sem94}.

Kostant's construction has been independently generalized by Etingof in \cite{Eti99} and by Sevostyanov in \cite{Sev99} to the case where the group $G$ was replaced by the corresponding quantized enveloping algebra $U_q(\mathfrak{g})$. The quantum integrable system constructed in this fashion is known as the $q$-\emph{deformed open Toda chain.} The present paper continues the study of the analytic theory of the $q$-deformed $\mathfrak{gl}_n$ Toda chain that was initiated by Kharchev, Lebedev, and Semenov-Tian-Shanksy in~\cite{KLS02}. One notable feature of this theory is its {\em modular duality}: if $q = e^{\pi i b^2}$ with $b\in \mathbb{R}_{>0}$, the eigenfunctions constructed in \cite{KLS02}, which we shall refer to here as the $b$\,-\emph{Whittaker functions}, are invariant under the exchange of~$b$ and~$b^{-1}$.  Based on this observation, it was suggested in \cite{KLS02} that the $b$\,-Whittaker functions should find their representation-theoretic meaning in terms of the \emph{modular double} of the split real quantum group $U_q(\mathfrak{g},\mathbb{R})$  introduced by Faddeev in \cite{Fad99} for $\mathfrak{g}=\mathfrak{sl}_2$, and later in higher rank by Frenkel and Ip in \cite{FI13, Ip12a, Ip15}. 

Recently, we have shown in \cite{SS17} that the $b$\,-Whittaker functions do indeed play a crucial role in the representation theory of $U_q(\mathfrak{sl}_n,\mathbb{R})$, and in fact govern the decomposition of a tensor product $\mathcal{P}_\la\otimes\mathcal{P}_\mu$ of two of its positive representations into (a direct integral of) irreducibles. The key observation from \cite{SS17} is that the operators realizing the diagonal action of
the fundamental Casimirs of $U_q(\mathfrak{sl}_n,\mathbb{R})$ on $\mathcal{P}_\la\otimes\mathcal{P}_\mu$ can be identified with the $q$-deformed  $\mathfrak{gl}_n$ open Toda Hamiltonians. As a consequence, the proof of the conjecture made by Frenkel and Ip in \cite{FI13} that the category of positive representations of  $U_q(\mathfrak{sl}_n,\mathbb{R})$ is closed under tensor product is reduced to proving the following theorem, which is the main result of the present paper.

\begin{theorem}
\label{main-thm}
Let $\Psi^{(n)}_{\bs\la}(\bs x)$ be the $b$\,-Whittaker function for the $q$-deformed $\gl_n$ Toda system as in Definition~\ref{def-whit}.
\begin{enumerate}
\item
The $b$\,-Whittaker transform $\Wc$, defined on the space of rapidly decaying test functions by the formula
$$
\Wc[f](\bs\la) = \int_{\R^n}f(\bs x)\overline{\Psi^{(n)}_{\bs\la}(\bs x)}d\bs x,
$$
extends to a unitary equivalence 
$$
\Wc\colon L^2(\R^n) \to L_{\mathrm{sym}}^2(\R^n,m(\bs\la)),
$$
where the target is the Hilbert space of symmetric functions in $\bs\la$ that are square-integrable with respect to the Sklyanin measure $m(\bs\la)$ defined by formula~\eqref{sklyanin-mes}. \\ [-9pt]
		
\item The $b$\,-Whittaker transform $\Wc$ intertwines the action of the $\gl_n$ $q$-deformed open Toda Hamiltonian $H^{(n)}_k$ on the Fock--Goncharov Schwartz space $\mathcal{S}(\R^n)$ with the operator of multiplication by the $k$-th elementary symmetric function $e_k(\bs \la)$ in variables $e^{2\pi b\la_j}$, where $j=1, \dots, n$.
\end{enumerate}
\end{theorem}

Theorem~\ref{main-thm} is also related to an important problem in quantum higher Teichm\"uller theory, see~\cite{FG06,FG09}. Recall that the results of \cite{SS17} were obtained using a construction of the positive representations of $U_q(\mathfrak{sl}_n,\mathbb{R})$  from representations of quantum cluster algebras associated to moduli spaces of framed $PGL_n$-local systems on a marked surface. In this context, we showed that the $q$-deformed quantum Toda Hamiltonians correspond to the operators quantizing the (elementary symmetric functions of the) eigenvalues of a local system's monodromy around a simple closed curve. As such, Theorem~\ref{main-thm} presents the key analytic ingredient in proving the {\em modular functor conjecture} of Fock and Goncharov, see \cite{FG09}, which describes how the mapping class group representations obtained from quantum higher Teichm\"uller theory behave under cutting and gluing of surfaces. Thus, completing the proof of the conjecture of Frenkel and Ip on tensor product decomposition on the one hand, and the modular functor conjecture on the other, serves as our primary motivation for the present work.

Our approach to proving Theorem~\ref{main-thm} and to developing the theory of the $q$-deformed Toda system is based on the structure of the latter as a cluster integrable system.  At the semi-classical ({\it i.e.} Poisson-geometric) level, this cluster structure has been investigated by various authors, see for example \cite{HKKR00, GK11, GSV13, FM16}. In this work, we show that it also plays a crucial role in the analysis of the quantum system.  More specifically, we explain in Section~\ref{sect-Ham} that the {\em Baxter $Q$-operator} of the $q$-deformed $\mathfrak{gl}_n$ open Toda system, which serves as a generating object for the Toda Hamiltonians, can be expressed as a sequence of quantum cluster mutations. Based on our cluster description of the Baxter operator, we obtain in Proposition~\ref{giv-prop} a new formula for the $b$\,-Whittaker functions that can be regarded as a modular $b$\,-deformation of Givental's formula \cite{Giv97} for the undeformed $\gl_n$-Whittaker functions. The advantage of the cluster formalism is that the essential properties of the Baxter operators and the $b$\,-Whittaker functions are formal consequences of the fundamental \emph{pentagon identity} satisfied by the modular quantum dilogarithm, and can thus be proved in the spirit of Volkov's ``noncommutative hypergeometry'', see~\cite{Vol05}.

Using similar techniques, we show in Proposition~\ref{dehn-eigen-prop} that the $b$\,-Whittaker functions are eigenfunctions of Dehn twist operator from quantum higher Teichm\"uller theory, which serves as one of the main steps in the proof of the modular functor conjecture. Proposition~\ref{dehn-eigen-prop} can be regarded as a generalization of the result of Kashaev \cite{Kas01} in the $\mathfrak{gl}_2$ case.

From our Givental-type formula in Proposition~\ref{giv-prop} for the $b$\,-Whittaker functions, we also derive in Propositions~\ref{orthog-int} and \ref{cauchy-littlewood} modular $b$\,-analogs of several integral identities for the undeformed Whittaker functions that were proven by Stade in \cite{Sta01, Sta02} as part of his work on the Rankin--Selberg method. The identity in Proposition~\ref{cauchy-littlewood} can be considered as a modular $b$\,-analog of the ``Cauchy--Littlewood'' identity for undeformed Whittaker functions derived by Corwin, O`Connell, Sepp\"al\"ainen and Zygouras in \cite{COSZ14}.

With these results in hand, we are able to obtain in Section~\ref{thm-proof} the two key ingredients in the proof of Theorem~\ref{main-thm}:  the \emph{completeness relation} $\Wc^*\Wc = \Id$ and the \emph{orthogonality relation} $\Wc\Wc^*  = \Id$. These relations are proved in Sections~\ref{sect-completeness} and \ref{sect-orthog} respectively. We conclude the article by presenting an application of our results to derive certain hypergeometric integral identities due to Gustafson \cite{Gus94} and Rains \cite{Rai09,Rai10}, based on the properties of the $b$\,-Whittaker functions.

Shortly after this paper was completed, an article~\cite{DKM18} by Derkachev, Kozlowski, and Manashov appeared in which the unitarity of the separation of variables transform for the modular {\it XXZ} spin chain and Sinh-Gordon model was established via different methods. The latter work builds on an approach to the {\it XXX} spin system developed in an influential series of papers~\cite{DKM01, DKM03, DM14}. It would be interesting to investigate whether their results can also be understood in cluster-theoretic terms.

\section*{Acknowledgements}

We are very grateful to Sergey Derkachov, Vladimir Fock, Igor Frenkel, Michael Gekhtman, Alexander Goncharov, Ivan Ip, Thomas Lam, Nicolai Reshetikhin, Michael Semenov-Tian-Shansky, Michael Shapiro, J\"org Teschner, and Oleksandr Tsymbaliuk for many helpful discussions, explanations, and suggestions. The second author has been supported by the NSF Postdoctoral Fellowship DMS-1703183 and by the RFBR grant 17-01-00585.

\section{Non-compact quantum dilogarithms}

In this section, we fix our conventions regarding the non-compact quantum dilogarithm function and recall some of its important properties. We also recall its close cousin, the so-called $c$-function, which often allows to make cumbersome formulas involving the non-compact quantum dilogarithm slightly more compact. For the rest of the paper we set
$$
c_b = \frac{i(b + b^{-1})}{2} \qquad\text{and}\qquad \Delta_b = \frac{i(b - b^{-1})}{2} \qquad\text{with}\qquad b \in \R_{>0}.
$$

\subsection{The non-compact quantum dilogarithm}

\begin{defn}
Let $C$ be the contour going along the real line from $-\infty$ to $+\infty$, surpassing the origin in a small semi-circle from above. The \emph{non-compact quantum dilogarithm function} $\varphi_b(z)$ is defined in the strip $\hm{\Im(z)} < \hm{\Im(c_b)}$ by the following formula \cite{Kas01}:
$$
\varphi_b(z) = \exp\hr{\frac{1}{4} \int_C\frac{e^{-2izt}}{\sinh(tb)\sinh(tb^{-1})}\frac{dt}{t}}.
$$
\end{defn}

The non-compact quantum dilogarithm can be analytically continued to the entire complex plane as a meromorphic function with an essential singularity at infinity. The resulting function $\varphi_b(z)$ enjoys the following properties \cite{Kas01}:

\begin{itemize}

\item[]{\bf poles and zeros:}
$$
\varphi_b(z)^{\pm1} = 0 \quad\Leftrightarrow\quad z = \mp\hr{c_b + ibm + ib^{-1}n} \quad\text{for}\quad m,n \in \Z_{\ge 0};
$$

\item[]{\bf behavior around poles and zeros:}
$$
\varphi_b(z\pm c_b) \sim \pm \zeta^{-1} (2 \pi i z)^{\mp1} \qquad\text{as}\qquad z \to 0;
$$

\item[]{\bf asymptotic behavior:}
$$
\varphi_b(z) \sim
\begin{cases}
\zeta_{\inv} e^{\pi i z^2}, & \Re(z) \to +\infty, \\
1, & \Re(z) \to -\infty,
\end{cases}
$$
where
$$
\zeta = e^{\pi i(1-4c_b^2)/12} \qquad\text{and}\qquad \zeta_{\inv} = \zeta^{-2} e^{-\pi i c_b^2};
$$

\item[]{\bf symmetry:}
$$
\varphi_b(z) = \varphi_{-b}(z) = \varphi_{b^{-1}}(z);
$$

\item[]{\bf inversion formula:}
\beq
\label{inv}
\varphi_b(z) \varphi_b(-z) = \zeta_{\inv} e^{\pi i z^2};
\eeq

\item[]{\bf functional equations:}
\beq
\label{eqn-func}
\varphi_b\hr{z - ib^{\pm1}/2} = \hr{1 + e^{2\pi b^{\pm1}z}} \varphi_b\hr{z + ib^{\pm1}/2};
\eeq

\item[]{\bf unitarity:}
$$
\overline{\varphi_b(z)} \varphi_b(\overline{z}) = 1;
$$

\item[]{\bf pentagon identity:}
Given any pair of self-adjoint operators $p$ and $x$ satisfying $[p,x] = \frac{1}{2\pi i}$ we have
\beq
\varphi_b(p) \varphi_b(x) = \varphi_b(x) \varphi_b(p+x) \varphi_b(p).
\label{pentagon}
\eeq
\end{itemize}

In what follows we will drop the subscript $b$ from the notation for the quantum dilogarithm, and simply write $\varphi(z)$. The functional equations~\eqref{eqn-func} yield the following simple lemma that we will use frequently in the sequel.

\begin{lemma}
\label{triv}
For self-adjoint operators $p$ and $x$ satisfying $[p,x] = \frac{1}{2\pi i}$, we have
\begin{align*}
\varphi(x)^{-1} e^{2\pi b p} \varphi(x) &= e^{2 \pi b p} + e^{2 \pi b (p+x)}, \\
\varphi(p) e^{2\pi b x} \varphi(p)^{-1} &= e^{2 \pi b x} + e^{2 \pi b (p+x)}.
\end{align*}
\end{lemma}

\begin{remark}
Note that
$$
q e^{2\pi bp} e^{2\pi bx} = e^{2\pi b(p+x)} = q^{-1} e^{2\pi bx} e^{2\pi bp}.
$$
\end{remark}

\subsection{Integral identites for $\varphi(z)$.}
The quantum dilogarithm function $\varphi(z)$ satisfies many important integral identities. Before describing some of them, let us fix a useful convention.
\begin{notation}
\label{contour-convention}
Throughout the paper, we will often consider contour integrals of the form 
$$
\int_{C} \prod_{j,k}\frac{\varphi(t-a_j)}{\varphi(t-b_k)}f(t)dt,
$$
where $f(t)$ is some entire function. Unless otherwise specified, the contour $C$ in such an integral is always chosen to be passing below the poles of $\varphi(t-a_j)$ for all $j$, above the poles of $\varphi(t-b_k)^{-1}$ for all $k$, and escaping to infinity in such a way that the integrand is rapidly decaying. 
\end{notation}

The Fourier transform of the quantum dilogarithm can be calculated explicitly by the following integrals:
\begin{align}
\label{Fourier-1}
\zeta \varphi(w) &= \int\frac{e^{2\pi i x(w-c_b)}}{\varphi(x-c_b)} dx, \\
\label{Fourier-2}
\frac{1}{\zeta \varphi(w)} &= \int \frac{\varphi(x+c_b)}{e^{2\pi i x(w+c_b)}} dx.
\end{align}
Note that in accordance with Notation~\ref{contour-convention}, the integration contours in~\eqref{Fourier-1} and~\eqref{Fourier-2} can be taken to be $\R+i0$ and $\R-i0$ respectively. From this point onwards, we write $\R\pm i0$ instead of $\R\pm i\eps$ with a sufficiently small positive number $\eps$.

It was shown in~\cite{FKV01} that the pentagon identity~\eqref{pentagon} is equivalent to either of the following integral analogs of Ramanujan's $_1\psi_1$ summation formula:
\begin{align}
\label{beta-1}
\frac{\varphi(a) \varphi(w)}{\varphi(a+w-c_b)} &= \zeta^{-1} \int \frac{\varphi(x+a)}{\varphi(x-c_b)} e^{2\pi i x(w-c_b)} dx, \\
\label{beta-2}
\frac{\varphi(a+w+c_b)}{\varphi(a) \varphi(w)} &= \zeta \int \frac{\varphi(x+c_b)}{\varphi(x+a)} e^{-2\pi i x(w+c_b)} dx.
\end{align}

In what follows we also make use of certain distributional identities that arise as singular limits of the above integrals.

\begin{lemma}
\label{dist-identities}
For any test function $f$ analytic in the strip of width $b$ above the real axis, we have
\begin{align*}
&\int \frac{\varphi(t-z+c_b)}{\varphi(s-z-c_b)} e^{4\pi ic_bz}f(t) dzdt = e^{4\pi ic_bs} f(s), \\
&\int \frac{\varphi(t-z+c_b)}{\varphi(s-z-c_b)} e^{2\pi i(2c_b+ib)z} f(t)dz dt= e^{2\pi i(2c_b+ib)s} \hr{f(s) - q^{-2}f(s+ib)}.
\end{align*}
\end{lemma}

\begin{proof}
We shall only give a proof of the second formula here, the proof of the first one is completely analogous. Shifting the integration variable $z \to z+s$ and consecutively changing its sign we see that the left hand side of the second identity becomes
$$
e^{2\pi i(2c_b+ib)s} \int \frac{\varphi(z+t-s+c_b)}{\varphi(z-c_b)} e^{-2\pi i(2c_b+ib)z} f(t)dzdt.
$$
Using the symmetry between $a$ and $w$ in the formula~\eqref{beta-1}, we turn the above expression into
$$
e^{2\pi i(2c_b+ib)s} \int \frac{\varphi(z-c_b-ib)}{\varphi(z-c_b)} e^{2\pi i(t-s)z} f(t)dzdt.
$$
We then apply functional identity~\eqref{eqn-func} to write
$$
e^{2\pi i(2c_b+ib)s} \int \hr{1-q^{-2}e^{2\pi bz}} e^{2\pi i(t-s)z} f(t)dz dt.
$$
Now, the result follows from Fourier inversion formula.
\end{proof}

\subsection{The $c$-function}
\label{subsec-c}
It will prove useful to introduce a relative of the quantum dilogarithm, the \emph{$c$-function,} which was considered in~\cite{KLS02} as a $q$-analogue of the Harish-Chandra function that controls the coordinate asymptotics of the undeformed Whittaker functions. The $c$-function is defined in the strip $0 < \Im(z) < 2\Im(c_b)$ by the formula
$$
c(z) = \exp\hr{-\PV \int_\R \frac{e^{-izt}}{\hr{e^{b t}-1}\hr{e^{b^{-1}t}-1}}\frac{dt}{t}}.
$$
Here $\PV\int_{\R}$ stands for the \emph{principal value} of the singular integral: the integrand has a pole at the origin, and the principal value is defined to be the average of the integrals over a pair of contours following the real line but bypassing the origin on either side. Similarly to the quantum dilogarithm $\varphi(z)$, the function $c(z)$ admits a meromorphic continuation to the entire complex plane with an essential singularity at infinity.

It is a simple matter to check that the functions $\varphi(z)$ and $c(z)$ are related by
$$
\varphi(z) = \zeta^{-1} c(z+c_b)^{-1} e^{\frac{\pi i}{2} (z^2-c_b^2)}.
$$
The latter equality implies the following properties of $c(z)$:

\begin{itemize}

\item[]{\bf Poles and zeros:}
$$
c(z)^{\pm1} = 0 \quad\Leftrightarrow\quad z = c_b \pm (c_b+ imb + inb^{-1}) \quad\text{for}\quad m,n \in \Z_{\ge 0};
$$

\item[]{\bf Behavior around poles:}

$$
c(z) \sim -(2\pi i z)^{-1} \qquad\text{as}\qquad z \to 0.
$$

\item[]{\bf Asymptotic behaviour:}

$$
c(z) \sim \zeta^{\pm1} e^{\mp\frac{\pi i}{2} z(z-2c_b)}
\qquad\text{as}\qquad \Re(z) \to \pm\infty.
$$

\item[]{\bf Inversion formula:}
$$
c(z)c(2c_b - z) = 1.
$$

\item[]{\bf Functional equations:}
$$
c\hr{z + ib^{\pm1}} = \hr{e^{-\pi b^{\pm1} z} - e^{\pi b^{\pm1} z}} ic(z).
$$

\item[]{\bf Complex conjugation:}
$$
\overline{c(z)} = c(-\overline z).
$$

\end{itemize}

\subsection{Sklyanin measure}

\begin{notation}
\label{rebus}
Throughout the rest of the paper, we will employ the following vector notations. Boldface letters shall stand for vectors, e.g. $\bs x = (x_1, \dots, x_n)$. Given such a vector $\bs x$, we denote the sum of its coordinates by
$$
\underline{\bs x} = x_1 + \dots + x_n.
$$
We also make use of the ``Russian rebus'' convention, and denote the vectors obtained by deleting the first and last coordinates in $\bs x$ by
\begin{align*}
\bs{x'} &= (x_1, \dots, x_{n-1}), & \bs{'x} &= (x_2, \dots, x_n), \\
\bs{x''} &= (x_1, \dots, x_{n-2}), & \bs{''x} &= (x_3, \dots, x_n),
\end{align*}
et cetera. Finally, we set
\beq
\label{rho}
\rho_s(\bs x) = \frac{1}{2} \sum_{1 \le j < k \le s} (x_j - x_k).
\eeq
\end{notation}

\begin{defn}
For $\bs\la \in \R^n$, we define the Sklyanin measure on $\R^n$ to be
\beq
\label{sklyanin-mes}
m(\bs\la)d\bs\la = \frac{1}{n!}\prod_{j\neq k}^n\frac{1}{c(\la_j-\la_k)}d\bs\la.
\eeq
\end{defn}

Note that by the inversion formula for the $c$-function we have 
$$
\frac{1}{c(\la)c(-\la)} = 4 \sinh(\pi b \la) \sinh(\pi b^{-1} \la),
$$
so that

$$
m(\bs\la) = \frac{1}{n!}\prod_{j<k} \hr{e^{\pi b(\la_k-\la_j)} - e^{\pi b(\la_j-\la_k)}} \hr{e^{\pi b^{-1}(\la_k-\la_j)} - e^{\pi b^{-1}(\la_j-\la_k)}}.
$$

The Sklyanin measure $m(\bs\la)$ can therefore be expanded into a linear combination of exponentials in the variables $\bs\la$. The function $m(\bs\la)$ is manifestly symmetric in the variables $\bs\la$. In fact, we have

\begin{prop}
The Sklyanin measure $m(\bs \la)$ on $\R^{n}$ can be written as the symmetrization
$$
m(\bs\la) = \frac{1}{n!}\sum_{\sigma \in {S}_n}\mf\hr{\sigma(\bs\la)},
$$
where
\beq
\label{non-symmetric-measure}
\mf(\bs\la) = \prod_{j=1}^n e^{(2j-n-1)\pi b \la_j} \prod_{j<k}^n \hr{e^{\pi b^{-1}(\la_k-\la_j)} - e^{\pi b^{-1}(\la_j-\la_k)}}.
\eeq
\end{prop}

\begin{proof}
The Proposition follows from using the Weyl denominator formula to express
$$
\prod_{j<k}^n \hr{e^{\pi b^{\pm1} (\la_k-\la_j)} - e^{\pi b^{\pm1}(\la_j-\la_k)}}=\prod_{j=1}^n e^{- \pi (n+1) b^{\pm1} \la_j} \sum_{\sigma \in {S}_n} (-1)^{\sigma} \prod_{k=1}^n e^{2\pi b^{\pm 1} k \la_{\sigma(k)}}.
$$
\end{proof}

\begin{cor}
\label{cor-measure}
The following recursive formula holds
$$
\mf(\bs\la) = e^{4\pi ic_b \rho_n(\bs\la)} \prod_{k=1}^n \frac{\varphi(\la_k-\la_{n+1}-\Delta_b)}{\varphi(\la_k-\la_{n+1}+c_b)} \mf(\bs\la').
$$
\end{cor}

\begin{proof}
We use formula~\eqref{non-symmetric-measure} to write
$$
\mf(\bs\la) = \prod_{j=1}^n e^{2\pi ic_b (\la_j-\la_{n+1})} \prod_{k=1}^n \hr{1 - e^{2\pi b^{-1} (\la_k-\la_{n+1})}} \mf(\bs\la').
$$
Now the result follows from the functional equation~\eqref{eqn-func}.
\end{proof}

\section{Quantum cluster mutations}

In this section we recall a few basic facts about cluster tori and their quantization following~\cite{FG09}. We shall only need the quantum cluster algebras related to quantum groups of type A, and we incorporate this in the definition of a cluster seed.

\begin{defn}
A cluster seed is a datum $\Theta=\hr{\La, (\cdot,\cdot),\hc{e_i}}$ where
\begin{itemize}
\item $\La$ is a lattice;
\item $(\cdot,\cdot)$ is a skew-symmetric $\Z$-valued form on $\La$;
\item $\hc{e_i \,|\, i \in I}$ is a basis of the lattice $\La$.
\end{itemize}
\end{defn}

\begin{remark}
The above definition differs from the general one by assuming that there are no frozen variables, and by setting all multipliers $d_i=1$.
\end{remark}

To a seed $\Theta$, we can associate a quiver $\mathcal{Q}$ with vertices labelled by the set $I$ and arrows given by the adjacency matrix $\eps = \hr{\eps_{ij}}$, such that $\eps_{ij} = (e_i,e_j)$. It is clear, that the seed can be restored from a quiver. Indeed, a vertex $i \in I$ corresponds to the basis vector $e_i$, which gives rise to a lattice as $i$ runs through $I$, while the adjacency matrix of the quiver defines the form $(\cdot, \cdot)$. 
 
The pair $\hr{\Lambda,(\cdot, \cdot)}$ determines a \emph{quantum torus algebra} $\mathcal{T}_\Lambda$, which is the free $\Z[q^{\pm1}]$-module spanned by $X_{\lambda}$, $\lambda\in \Lambda$, with $X_0 = 1$ and the multiplication defined by
$$
q^{(\lambda,\mu)}X_\lambda X_\mu = X_{\lambda+\mu}.
$$
A basis $\hc{e_i}$ of the lattice $\La$ gives rise to a distinguished system of generators for $\mathcal{T}_\Lambda$, namely the elements $X_i=X_{e_i}$. Given the quantum torus algebra $\mathcal{T}_\Lambda$ one can consider an associated Heisenberg $*$-algebra $\mathcal{H}_\Lambda$. It is a topological $*$-algebra over $\C$ generated by elements $\{x_i\}$ satisfying
$$
[x_j,x_k]=\frac{1}{2\pi i} \eps_{jk} \qquad\text{and}\qquad *x_j=x_j. 
$$
Then the assignments
$$
X_j = e^{2\pi b x_j} \qquad\text{and}\qquad q=e^{\pi i b^2}
$$
define a homomorphism of algebras $\mathcal{T}_\La \hookrightarrow \mathcal{H}_\La$.

Let $\Theta$ be a seed, and $k \in I$ a vertex of the corresponding quiver $\mathcal{Q}$. Then one obtains a new seed, $\mu_k(\Theta)$, called the \emph{mutation of $\Theta$ in direction $k$}, by changing the basis~$\hc{e_i}$ while the rest of the data remains the same. The new basis $\{e_i'\}$ is
$$
e'_i = 
\begin{cases}
-e_k &\text{if} \; i=k, \\
e_i + [\eps_{ik}]_+e_k &\text{if} \; i \ne k,
\end{cases}
$$
where $[a]_+=\max(a,0)$. We remark that bases the $\hc{e_i}$ and $\hc{\mu_k^2(e_i)}$ do not necessarily coincide, although the seeds $\Theta$ and $\mu_k^2(\Theta)$ are isomorphic.

For each mutation $\mu_k$ we define an algebra automorphism of the skew field $\Frac(\mathcal{T}_\Lambda)$:
$$
\mu_k = \Ad_{\varphi^{-1}(-x_k)}.
$$
By abuse of notation we call this automorphism a \emph{quantum mutation} and denote it by the same symbol $\mu_k$. The fact that conjugation by $\varphi^{-1}(-x_k)$ yields a genuine birational automorphism of $\mathcal{T}_\La$ is guaranteed by the integrality of the form~$(\cdot, \cdot)$ and functional equations~\eqref{eqn-func}. For example, the statement of the Lemma~\ref{triv} is equivalent to
$$
\mu_k\hr{X_{e_i}} =
\begin{cases}
X_{e'_i}\hr{1+qX_{e'_k}} &\text{if} \quad \eps_{ki}=1, \\
X_{e'_i}\hr{1+qX_{e'_k}^{-1}}^{-1} &\text{if} \quad \eps_{ki}=-1.
\end{cases}
$$

\section{Baxter $Q$-operators and their cluster realization}
\label{baxter-section}
\begin{defn}
\label{quiver-def}
For $n\geq 2$, we define a quiver $\Qc_n$ with $2n$ vertices $\hc{v_0, \dots, v_{2n-1}}$ and the following arrows:
\begin{itemize}
\item for all $1 \le k \le n-1$ there is a double arrow $v_{2k} \Rightarrow v_{2k-1}$;
\item for all $1 \le k \le n$ there is an arrow $v_{2k-1} \rightarrow v_{2k-2}$;
\item for all $1 \le k \le n-2$ there is an arrow $v_{2k-1} \rightarrow v_{2k+2}$;
\item there are additional arrows $v_0 \rightarrow v_2$ and $v_{2n-3}\rightarrow v_{2n-1}$.
\end{itemize}
\end{defn}

\begin{figure}[h]
\begin{tikzpicture}[every node/.style={inner sep=0, minimum size=0.5cm, thick, draw, circle}, x=0.75cm, y=0.5cm]

\node (1) at (0,6) {\footnotesize{$0$}};
\node (2) at (-1,4) {\footnotesize{$1$}};
\node (3) at (1,4) {\footnotesize{$2$}};
\node (4) at (-1,2) {\footnotesize{$3$}};
\node (5) at (1,2) {\footnotesize{$4$}};
\node (6) at (-1,0) {\footnotesize{$5$}};
\node (7) at (1,0) {\footnotesize{$6$}};
\node (8) at (0,-2) {\footnotesize{$7$}};

\draw [<-, thick] (1) -- (2);
\draw [<-, thick] (3) -- (1);
\draw [<-, thick] (3) -- (4);
\draw [<-, thick] (5) -- (2);
\draw [<-, thick] (5) -- (6);
\draw [<-, thick] (7) -- (4);
\draw [<-, thick] (8) -- (6);
\draw [<-, thick] (7) -- (8);

\draw [<-, thick] (2.10) -- (3.-190);
\draw [<-, thick] (2.-10) -- (3.190);
\draw [<-, thick] (4.10) -- (5.-190);
\draw [<-, thick] (4.-10) -- (5.190);
\draw [<-, thick] (6.10) -- (7.-190);
\draw [<-, thick] (6.-10) -- (7.190);

\end{tikzpicture}
\caption{The quiver $\mathcal{Q}_4$.}
\label{fig-cox1}
\end{figure}
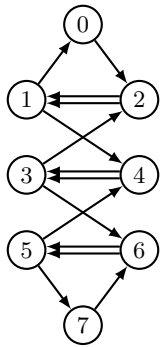

Consider the quiver $\Qc_n$, see Figure~\ref{fig-cox1} for the $n=4$ case, and the corresponding cluster seed together with its Heisenberg $*$-algebra $\Hc_n$. We will abuse notation once again and denote the generators of the Heisenberg algebra by the same symbols $\hc{v_0, \dots, v_{2n-1}}$ as the vertices of the Coxeter quiver. 

For any real numbers $u$ and $v$, the assignment
$$
v_0 \mapsto -p_1-u, \qquad v_{2j-1} \mapsto x_j - x_{j+1} + p_j - p_{j+1}, \qquad v_{2j} \mapsto x_{j+1} - x_j, \qquad v_{2n-1} \mapsto p_n+v
$$
where $j=1, \dots, n-1$  defines a representation of $\Hc_n$ on the Hilbert space of $L^2$ functions in the variables $x_1, \dots, x_n$, in which the generators $v_j$ of $\Hc_n$ act by unbounded self-adjoint operators. 

\begin{defn}
We define the (top) \emph{Baxter $Q$-operator} $Q^{\mathrm{t}}_n(u)$ to be the operator whose inverse is the composition of consecutive mutation operators $\mu_k^\sharp$ at vertices $v_0, v_1, \dots, v_{2n-2}$ in the quiver $\Qc_n$:
$$
Q^{\mathrm{t}}_n(u)^{-1} = \mu^\sharp_{2n-2} \dots \mu^\sharp_1 \mu^\sharp_0.
$$
Similarly, we define the (bottom) Baxter operator $Q^{\mathrm{b}}_n(u)$ to be the operator whose inverse is the composition of consecutive mutation operators at vertices $v_{2n-1}, v_{2n-3},v_{2n-2} \dots, v_{1},v_2$ in the quiver $\Qc_n$:
$$
Q^{\mathrm{b}}_n(u)^{-1} = \mu^\sharp_{2}\mu^\sharp_1 \dots \mu^\sharp_{2n-2} \mu^\sharp_{2n-3} \mu^\sharp_{2n-1}.
$$

\end{defn}

\begin{prop}
We have
$$
Q^{\mathrm{t}}_1(u) = \varphi(p_1+u) \qquad\text{and}\qquad Q^{\mathrm{t}}_n(u) = Q^{\mathrm{t}}_1(u)T^{\mathrm{t}}_2(u) \dots T^{\mathrm{t}}_n(u)
$$
where
$$
T^{\mathrm{t}}_k(u) = \varphi(p_k+x_k-x_{k-1}+u)\varphi(p_k+u).
$$
Similarly, setting
$$
T^{\mathrm{b}}_k(v) = \varphi(-p_k-v)\varphi(-p_{k-1}+x_k-x_{k-1}-v),
$$
we have
$$
Q_{1}^{\mathrm{b}}(v) = \varphi(-p_1-v) \qquad\text{and}\qquad Q_{n}^{\mathrm{b}}(v) = T^{\mathrm{b}}_{n}(v)\dots T^{\mathrm{b}}_{2}Q_{1}^{\mathrm{b}}(v).
$$
\end{prop}

\begin{proof}
The result is immediate from the above definitions.
\end{proof}

At the level of quivers, the effect of the top Baxter operator mutation sequence 
 is to transport the ``handle'' vertex $v_0$ at the top of the Coxeter quiver down to create a new handle at the bottom of the quiver. The quiver $\Qc_n'$ obtained from $\Qc_n$ by applying this sequence is illustrated in Figure~\ref{fig-cox2}.
\begin{figure}[h]
\begin{tikzpicture}[every node/.style={inner sep=0, minimum size=0.5cm, thick, draw, circle}, x=0.75cm, y=0.5cm]

\node (1) at (-1,-2) {\footnotesize{$6$}};
\node (8) at (1,-2) {\footnotesize{$7$}};
\node (2) at (-1,4) {\footnotesize{$0$}};
\node (3) at (1,4) {\footnotesize{$1$}};
\node (4) at (-1,2) {\footnotesize{$2$}};
\node (5) at (1,2) {\footnotesize{$3$}};
\node (6) at (-1,0) {\footnotesize{$4$}};
\node (7) at (1,0) {\footnotesize{$5$}};

\draw [<-, thick] (1) -- (6);
\draw [<-, thick] (1) -- (7);
\draw [<-, thick] (8) -- (6);
\draw [<-, thick] (8) -- (7);

\draw [<-, thick] (3) -- (4);
\draw [<-, thick] (5) -- (2);
\draw [<-, thick] (5) -- (6);
\draw [<-, thick] (7) -- (4);

\draw [<-, thick] (2.10) -- (3.-190);
\draw [<-, thick] (2.-10) -- (3.190);
\draw [<-, thick] (4.10) -- (5.-190);
\draw [<-, thick] (4.-10) -- (5.190);
\draw [<-, thick] (6.10) -- (7.-190);
\draw [<-, thick] (6.-10) -- (7.190);

\end{tikzpicture}
\caption{The quiver $\Qc'_4$.}
\label{fig-cox2}
\end{figure}

It is straightforward to see that the composite operator 
$$
Q^{\mathrm{swap}}_n(u,v) = Q^{\mathrm{b}}_n(v)Q^{\mathrm{t}}_n(u)
$$
can be realized as the inverse of the sequence of mutation operators

$$
\mu^{\mathrm{swap}} =   (\mu^{\sharp}_2\mu^{\sharp}_1)(\mu^{\sharp}_4\mu^{\sharp}_3)\dots(\mu^{\sharp}_{2n-2} \mu^{\sharp}_{2n-3})\mu^{\sharp}_{2n-1}\circ (\mu^{\sharp}_{2n-2} \dots \mu^{\sharp}_1 \mu^{\sharp}_0).
$$
For $n=4$, the quiver $\Qc_n'' $ obtained from $\Qc_n$ by applying the mutation sequence $\mu^{\mathrm{swap}}$ is illustrated in Figure~\ref{fig-cox-swap}. 

\begin{figure}[h]
\begin{tikzpicture}[every node/.style={inner sep=0, minimum size=0.5cm, thick, draw, circle}, x=0.75cm, y=0.5cm]

\node (1) at (0,6) {\footnotesize{$1$}};
\node (2) at (-1,4) {\footnotesize{$3$}};
\node (3) at (1,4) {\footnotesize{$0$}};
\node (4) at (-1,2) {\footnotesize{$5$}};
\node (5) at (1,2) {\footnotesize{$2$}};
\node (6) at (-1,0) {\footnotesize{$7$}};
\node (7) at (1,0) {\footnotesize{$4$}};
\node (8) at (0,-2) {\footnotesize{$6$}};

\draw [<-, thick] (1) -- (2);
\draw [<-, thick] (3) -- (1);
\draw [<-, thick] (3) -- (4);
\draw [<-, thick] (5) -- (2);
\draw [<-, thick] (5) -- (6);
\draw [<-, thick] (7) -- (4);
\draw [<-, thick] (8) -- (6);
\draw [<-, thick] (7) -- (8);

\draw [<-, thick] (2.10) -- (3.-190);
\draw [<-, thick] (2.-10) -- (3.190);
\draw [<-, thick] (4.10) -- (5.-190);
\draw [<-, thick] (4.-10) -- (5.190);
\draw [<-, thick] (6.10) -- (7.-190);
\draw [<-, thick] (6.-10) -- (7.190);

\end{tikzpicture}
\caption{The quiver $\mathcal{Q}''_4$.}
\label{fig-cox-swap}
\end{figure}

Observe that the quiver $\Qc_n'' $ is isomorphic to the original quiver $\Qc_n$ via the permutation
$$
v_0 \mapsto v''_{1}, \qquad v_{2j-1} \mapsto v''_{2j+1},\qquad v_{2j} \mapsto v''_{2j-2}, \qquad v_{2n-1} \mapsto v''_{2n-2}
$$
for all $1\leq j\leq n-1$. The mutated basis is represented as follows:
$$
v''_{1} = -p_1-v, \qquad v''_{2j+1} = p_{j}-p_{j+1} + x_j-x_{j+1}, \qquad v''_{2j-2} = x_{j+1}-x_j, \qquad v''_{2n-2} = p_n+u.
$$
Thus, the sequence of mutations $\mu^{\mathrm{swap}}$ has the effect of swapping the top and bottom handles of the quiver $\Qc_n$.  

\begin{remark}
The operator $Q_n^{\mathrm{swap}}(u,v)$ introduced above is a higher rank generalization of the operator $Q(u,v)$ that first appeared in \cite{Kas01} and was used in {\it op. cit.} to formulate a non-compact analog of the Bailey lemma from the theory of hypergeometric functions. 
\end{remark}

\subsection{Commutativity of Baxter operators}

We now explain how the commutativity of Baxter operators follows from the pentagon identity~\eqref{pentagon} satisfied by the quantum dilogarithm function. We establish commutativity of the top Baxter operators $Q^{\mathrm{t}}_n(u)$ and $Q^{\mathrm{t}}_n(v)$; the commutativity of the corresponding bottom Baxter operators can be given a completely analogous treatment. 

Let us introduce an operator
$$
S_{n+1}(u) = \varphi(p_{n+1}+p_n+x_{n+1}-x_{n-1}+u) \varphi(p_{n+1}+p_n+x_{n+1}-x_n+u).
$$

\begin{lemma}
\label{lem-TT}
The following equalities hold:
$$
T^{\mathrm{t}}_n(u)T^{\mathrm{t}}_n(v) = T^{\mathrm{t}}_n(v)T^{\mathrm{t}}_n(u)
\qquad\text{and}\qquad
T^{\mathrm{t}}_{n+1}(u)T^{\mathrm{t}}_n(v) = T^{\mathrm{t}}_n(v)S_{n+1}(u+v)T^{\mathrm{t}}_{n+1}(u).
$$
\end{lemma}

\begin{proof}
Both statements of the Lemma are direct consequences of the pentagon identity. For example, the first equality can be obtained as follows: setting
$$
\tilde u = u-x_{n-1} \qquad\text{and}\qquad \tilde v = v-x_{n-1}
$$
we get
\begin{align*}
T^{\mathrm{t}}_n(u)T^{\mathrm{t}}_n(v)
&= \varphi(p_n+x_n+\tilde u) \varphi(p_n+u) \varphi(p_n+x_n+\tilde v) \varphi(p_n+v) \\ 
&= \varphi(p_n+x_n+\tilde u) \varphi(p_n+x_n+\tilde v) \varphi(2p_n+x_n+u+\tilde v) \varphi(p_n+u) \varphi(p_n+v) \\ 
&= \varphi(p_n+x_n+\tilde v) \varphi(p_n+x_n+\tilde u) \varphi(2p_n+x_n+\tilde u+v) \varphi(p_n+v) \varphi(p_n+u) \\ 
&= \varphi(p_n+x_n+\tilde v) \varphi(p_n+v) \varphi(p_n+x_n+\tilde u) \varphi(p_n+u) \\
&= T^{\mathrm{t}}_n(v)T^{\mathrm{t}}_n(u).
\end{align*}
The second relation is proved in a similar fashion.
\end{proof}

\begin{lemma}
\label{lem-QQ}
The top and bottom Baxter operators form commutative families: for all $u,v$ we have
$$
Q^{\mathrm{t}}_{n}(u) Q^{\mathrm{t}}_{n}(v) = Q^{\mathrm{t}}_{n}(v) Q^{\mathrm{t}}_{n}(u),
$$
and
$$
Q^{\mathrm{b}}_{n}(u) Q^{\mathrm{b}}_{n}(v) = Q^{\mathrm{b}}_{n}(v) Q^{\mathrm{b}}_{n}(u).
$$
\end{lemma}

\begin{proof}
The Lemma is proved by induction over $n$. We show the commutativity of the top Baxter operators, the proof for the bottom ones being completely analogous. The commutativity is immediate for $n=1$, and the inductive step follows from Lemma~\ref{lem-TT}. Indeed, we have
\begin{align*}
Q^{\mathrm{t}}_{n+1}(u) Q^{\mathrm{t}}_{n+1}(v)
&= Q^{\mathrm{t}}_{n}(u)T^{\mathrm{t}}_{n+1}(u)Q^{\mathrm{t}}_{n-1}(v)T^{\mathrm{t}}_n(v)T^{\mathrm{t}}_{n+1}(v) \\
&= Q^{\mathrm{t}}_{n}(u)Q^{\mathrm{t}}_{n-1}(v)T^{\mathrm{t}}_{n+1}(u)T^{\mathrm{t}}_n(v)T^{\mathrm{t}}_{n+1}(v) \\
&= Q^{\mathrm{t}}_{n}(u)Q^{\mathrm{t}}_{n-1}(v)T^{\mathrm{t}}_n(v)S_{n+1}(u+v)T^{\mathrm{t}}_{n+1}(u)T^{\mathrm{t}}_{n+1}(v) \\
&= Q^{\mathrm{t}}_{n}(u)Q^{\mathrm{t}}_{n}(v)S_{n+1}(u+v)T^{\mathrm{t}}_{n+1}(u)T^{\mathrm{t}}_{n+1}(v) \\
&= Q^{\mathrm{t}}_{n}(v)Q^{\mathrm{t}}_{n}(u)S_{n+1}(u+v)T^{\mathrm{t}}_{n+1}(v)T^{\mathrm{t}}_{n+1}(u).
\end{align*}
The latter equality shows that the left hand side is symmetric in $u$ and $v$ and hence
$$
Q^{\mathrm{t}}_{n+1}(u) Q^{\mathrm{t}}_{n+1}(v) = Q^{\mathrm{t}}_{n+1}(v) Q^{\mathrm{t}}_{n+1}(u).
$$
\end{proof}

Moreover, it turns out that the top and bottom Baxter operators also commute with each other. This can be easily proved with the help of the following lemma.
\begin{lemma}
\label{tb-lem}
We have
$$
T^{\mathrm{b}}_{n+1}(v)T^{\mathrm{t}}_n(u)T^{\mathrm{b}}_n(v)T^{\mathrm{t}}_{n+1}(u) = T^{\mathrm{t}}_{n}(u)T^{\mathrm{t}}_{n+1}(u)T^{\mathrm{b}}_{n+1}(v)T^{\mathrm{b}}_{n}(v).
$$
\end{lemma}
\begin{proof}
Let us show that 
\beq
\label{db}
T^{\mathrm{t}}_{n}(u)^{-1}T^{\mathrm{b}}_{n+1}(v)T^{\mathrm{t}}_n(u)T^{\mathrm{b}}_n(v)T^{\mathrm{t}}_{n+1}(u) T^{\mathrm{b}}_{n}(v)^{-1} = T^{\mathrm{t}}_{n+1}(u)T^{\mathrm{b}}_{n+1}(v).
\eeq
It follows from a single application of the pentagon relation that the product of the first three factors in the left hand side of~\eqref{db} is equal to
\beq
\label{db-1}
\varphi(-p_{n+1}-v)\varphi(x_{n+1}-x_n+u-v)\varphi(-p_n+x_{n+1}-x_n-v).
\eeq
Similarly, the product of the last three factors is equal to
\beq
\label{db-2}
\varphi(p_{n+1}+x_{n+1}-x_n+u)\varphi(p_{n+1}-p_n+x_{n+1}-x_n+u-v)\varphi(p_{n+1}+u).
\eeq
Observe now that the rightmost factor in~\eqref{db-1} commutes with the leftmost factor in~\eqref{db-2}. We can thus rewrite the left hand side of~\eqref{db} as
\begin{align*}
\varphi(-p_{n+1}-v) &\varphi(x_{n+1}-x_n+u-v) \varphi(p_{n+1}+x_{n+1}-x_n+u) \\
\cdot &\varphi(-p_n+x_{n+1}-x_n-v) \varphi(p_{n+1}-p_n+x_{n+1}-x_n+u-v) \varphi(p_{n+1}+u).
\end{align*}
Applying the pentagon identity once again to the first and the last three factors of the above product we obtain the right hand side of~\eqref{db}.
\end{proof}

\begin{prop}
\label{tb-commut}
The following equality holds:
$$
Q_{n}^{\mathrm{t}}(u)Q_{n}^{\mathrm{b}}(v)=Q_{n}^{\mathrm{b}}(v)Q_{n}^{\mathrm{t}}(u).
$$
\end{prop}
\begin{proof}
The Proposition is proved by induction on $n$ with the help of Lemma~\ref{tb-lem}. The base case $n=1$ is immediate. For the inductive step, we write
\begin{align*}
Q_{n+1}^{\mathrm{b}}(v)Q_{n+1}^{\mathrm{t}}(u)&= T_{n+1}^{\mathrm{b}}(v)Q_{n}^{\mathrm{b}}(v)Q_{n}^{\mathrm{t}}(u)T_{n+1}^{\mathrm{t}}(u)\\
&= T_{n+1}^{\mathrm{b}}(v)Q_{n}^{\mathrm{t}}(u)Q_{n}^{\mathrm{b}}(v)T_{n+1}^{\mathrm{t}}(u)\\
&=Q_{n-1}^{\mathrm{t}}(u)T_{n+1}^{\mathrm{b}}(v)T_{n}^{\mathrm{t}}(u)T_{n}^{\mathrm{b}}(v)T_{n+1}^{\mathrm{t}}(u)Q_{n-1}^{\mathrm{b}}(v)\\
&=Q_{n-1}^{\mathrm{t}}(u) T^{\mathrm{t}}_{n}(u)T^{\mathrm{t}}_{n+1}(u)T^{\mathrm{b}}_{n+1}(v)T^{\mathrm{b}}_{n}(v)Q_{n-1}^{\mathrm{b}}(v)\\
&=Q_{n+1}^{\mathrm{t}}(u)Q_{n+1}^{\mathrm{b}}(v),
\end{align*}
where we used the inductive hypothesis and Lemma~\ref{tb-lem}. 
\end{proof}

\begin{notation}
In the remainder of the paper, we will mostly focus on the top Baxter operator $Q_{n}^{\mathrm{t}}(u)$. In order to lighten the notation, we will drop the superscript $\mathrm{t}$ and simply write $Q_{n}(u)$.
\end{notation}

\subsection{Relation with the Dehn twist operator}

We now recall the form of the operator representing the Dehn twist for an annulus in quantum higher Teichm\"uller theory. More presicely, in the setup of~\cite{SS17}, one can consider the cluster mapping class group element corresponding to the Dehn twist for the annulus with 1 marked point on each of its boundary components. This Dehn twist can be expressed in the coordinate chart $\Sigma_{n}^{\mathrm{cox}}$ defined in {\it op. cit.} as follows.

\begin{defn}
We define the $\gl_n$ Dehn twist operator to be the following cluster transformation of $\Qc_n$:
$$
D_n = \pi_{\mathrm{Dehn}}\circ \prod_{k=1}^{n-1} \mu_{2k-1},
$$
where the permutation $\pi_{\mathrm{Dehn}}$ is the product of transpositions $v_{2k} \leftrightarrow v_{2k-1}$ for $k = 1, \dots, n-1$.
\end{defn}

It is a straightforward matter to verify that the Dehn twist operator $D_n$ can be written explicitly as follows.

\begin{lemma}
The Dehn twist operator $D_n$ can be written as
$$
D_n = \prod_{j=1}^{n-1} \varphi(x_{j+1}-x_j)^{-1} \prod_{j=1}^n e^{-\pi i p_j^2}.
$$
\end{lemma}

%
%
%
%

Interestingly, the Dehn twist $D_n$ turns out to coincide with a degeneration of $Q_n^{\mathrm{swap}}(u,v)$.

\begin{prop}
\label{degen-dehn}
We have
$$
D_n = \zeta_{\inv}^n Q_n^{\mathrm{swap}}(0,0)^{-1}.
$$
\end{prop}

\begin{proof}
We prove the Proposition by induction over the rank $n$. The base case $n=1$ follows immediately from the inversion formula~\eqref{inv}, indeed
$$
Q_n^{\mathrm{swap}}(0,0) = Q_1^{\mathrm{b}}(0)Q_1^{\mathrm{t}}(0) = \varphi(-p_1)\varphi(p_1) = \zeta_{\inv} e^{\pi i p_1^2}.
$$
For the step of induction, we have
$$
Q_{n+1}^{\mathrm{swap}}(0,0) = T_{n+1}^{\mathrm{b}}(0) Q_{n}^{\mathrm{swap}}(0,0) T_{n+1}^{\mathrm{t}}(0) = \zeta_{\inv}^n T_{n+1}^{\mathrm{b}}(0) D_n^{-1} T_{n+1}^{\mathrm{t}}(0).
$$
Commuting $D_n^{-1}$ to the left we obtain the expression
$$
\zeta_{\inv}^n \prod_{j=1}^n e^{\pi i p_j^2} \varphi(-p_{n+1}) \prod_{j=1}^{n} \varphi(x_{j+1}-x_j) \varphi(p_{n+1}+x_{n+1}-x_n) \varphi(p_{n+1})
$$
Now, the pentagon identity applied to the product
$$
\varphi(x_{n+1}-x_n)\varphi(p_{n+1}+x_{n+1}-x_n)\varphi(p_{n+1})
$$
together with the inversion formula yield
\begin{align*}
Q_{n+1}^{\mathrm{swap}}(0,0) &= \zeta_{\inv}^n \prod_{j=1}^n e^{\pi i p_j^2} \varphi(-p_{n+1}) \varphi(p_{n+1}) \prod_{j=1}^{n}\varphi(x_{j+1}-x_j) \\
&= \zeta_{\inv}^{n+1} \prod_{j=1}^{n+1} e^{\pi i p_j^2} \prod_{j=1}^{n} \varphi(x_{j+1}-x_j) \\
&= \zeta_{\inv}^{n+1} D_{n+1}^{-1}.
\end{align*}
The Proposition is proved. 
\end{proof}
One can thus regard the operator $Q_n^{\mathrm{swap}}(u,v)$ as a 2-parametric deformation of the Dehn twist in quantum higher Teichm\"uller theory. 

\begin{remark}
Recently, several papers were devoted to the study of $q$-Toda systems and their spectral problems. In~\cite{GT18} the generalized quantum difference Toda lattices were considered, while in~\cite{BKP18a, BKP18b} the spectral problem for the periodic $q$-Toda system was investigated from the point of view of the quantum inverse scattering method. It would be interesting to give a cluster interpretation of these results.
\end{remark}



\section{Schwartz space and integral operators}

\subsection{Schwartz space}

Consider the space $\Fc$ consisting of entire functions $f$ such that
$$
\int_\R e^{xs} \hm{f(x + iy)}^2 dx < \infty \qquad\text{for all} \qquad s, y \in \R.
$$
The space $\Fc$ is dense in $L^2(\R)$, as can be seen from that fact that it contains the subspace
$$
\Fc_0  = \hc{e^{-\alpha x^2+\beta x}p(x) \,\big|\, \alpha\in \R_{>0}, \; \beta\in\C, \; p(x)\in\C[x]}.
$$

In the sequel, we shall also consider the higher-dimensional analogs $\Fc_{n} := \bigotimes_{k=1}^n\Fc$ of these spaces of test functions. Let us also recall an important analytic construction from \cite{Gon05,FG09}. Consider the algebra $\mathbb{L}_n$ consisting of universally Laurent elements in the quantum torus algebra $\mathcal{T}_\La$: that is, elements $A\in\mathcal{T}_\La$ such that for any sequence of mutations $\mu = \mu_{i_1} \circ \cdots \circ \mu_{i_k}$ one has $\mu(A)\in\mathcal{T}_\La$. 

\begin{defn}[\cite{FG09}]
The \emph{Fock--Goncharov Schwartz space} $\mathcal{S}$ is the subspace of $L^2(\mathbb{R}^n)$ consisting of all vectors $f\in L^2(\mathbb{R}^n)$ such that the functional on $\Fc_n$ defined by $w\mapsto \langle f,Aw\rangle$, $w\in\Fc_n$, is continuous in the $L^2$-norm for all $A\in\mathbb{L}_n$. 
\end{defn}
The Schwartz space $\mathcal{S}$ is the common domain of definition of the operators from $\mathbb{L}_n$.  It has a topology given by the family of semi-norms $\hn{Af}_{L^2}$ where $A$ runs over a basis in $\mathbb{L}_n$.

\subsection{Analytic continuation and integral operators}

Let us now consider the unitary operator $\varphi(p+u)$ on $L^2(\R)$. It can be written as an explicit integral operator on the space of test functions $\Fc$  with the help of formula~\eqref{Fourier-1} for the Fourier transform of the quantum dilogarithm. Indeed, we have
\begin{align*}
\varphi(p+u) f(x)
&= \zeta^{-1} \int \frac{e^{2\pi i t(p+u-c_b)}}{\varphi(t-c_b)} f(x) dt \\
&= \zeta^{-1} \int \frac{e^{2\pi i t(u-c_b)}}{\varphi(t-c_b)} f(x+t) dt.
\end{align*}
Shifting the integration variable $t \to t-x$, we obtain
$$
\varphi(p+u) f(x) = \zeta^{-1} e^{2\pi i x(c_b-u)} \int \frac{e^{2\pi i t(u-c_b)}}{\varphi(t-x-c_b)} f(t) dt.
$$

\begin{remark}
\label{analytic-continuation}
This latter formula allows us to define the action of $\varphi(p+u)$ on a test function when the parameter $u$ is no longer constrained to lie on the real line. In the sequel, all operators of the form $\varphi(a+u)$ where $a$ is self-adjoint and $u\notin \R$ are to be understood in this sense as operators on the appropriate space of test functions.
\end{remark}

\begin{lemma}
\label{Tf}
For $f(\bs x)\in\Fc_n$, we have
$$
T_k(u) f(\bs x) = \zeta^{-1} \int \frac{e^{2\pi i (y_k-x_k)(u-c_b)} \varphi(y_k-x_{k-1})}{\varphi(y_k-x_k-c_b) \varphi(x_k-x_{k-1})} f(\bs x)|_{x_k \to y_k} dy_k
$$
where
$$
f(\bs x)|_{x_k \to y_k} = f(x_1, \dots, x_{k-1},y_k,x_{k+1}, \dots, x_n).
$$
\end{lemma}

\begin{proof}
Applying the Fourier transform~\eqref{Fourier-1} twice we get
\begin{align*}
T_k(u) f(\bs x)
&= \zeta^{-2} \int  \frac{e^{2\pi it(p_k+x_k-x_{k-1}+u-c_b)} e^{2\pi is(p_k+u-c_b)}}{\varphi(t-c_b) \varphi(s-c_b)} f(\bs x) dt ds \\
&= \zeta^{-2} \int \frac{e^{2\pi i (t+s)(u-c_b)} e^{\pi i t^2} e^{2\pi it(x_k-x_{k-1})} e^{2\pi i(t+s)p_k}}{\varphi(t-c_b) \varphi(s-c_b)} f(\bs x) dt ds \\
&= \zeta^{-2} \int \frac{e^{2\pi i (t+s)(u-c_b)} e^{\pi i t^2} e^{2\pi it(x_k-x_{k-1})}}{\varphi(t-c_b) \varphi(s-c_b)} f(\bs x)|_{x_k \to x_k +t +s} dt ds.
\end{align*}
Now we change variables, setting $y_k = x_k+t+s$ to obtain
$$
T_k(u) f(\bs x) = \zeta^{-2} \int \frac{e^{2\pi i (y_k-x_k)(u-c_b)} e^{\pi i t^2} e^{2\pi it(x_k-x_{k-1})}}{\varphi(t-c_b) \varphi(y_k-x_k-t-c_b)} f(\bs x)|_{x_k \to y_k} dt dy_k.
$$
Changing the sign of the integration variable $t$ and applying the inversion formula for the quantum dilogarithm, we have
$$
T_k(u) f(\bs x) = \int \frac{\varphi(t+c_b)}{\varphi(t+y_k-x_k-c_b)} e^{-2\pi it(x_k-x_{k-1}+c_b)} e^{2\pi i (y_k-x_k)(u-c_b)} f(\bs x)|_{x_k \to y_k} dt dy_k.
$$
Finally, using formula~\eqref{beta-2} to take the integral over $t$ we obtain
\begin{align*}
T_k(u) f(\bs x) = \zeta^{-1} \int \frac{e^{2\pi i (y_k-x_k)(u-c_b)} \varphi(y_k-x_{k-1})}{\varphi(y_k-x_k-c_b) \varphi(x_k-x_{k-1})} f(\bs x)|_{x_k \to y_k} dy_k.
\end{align*}
\end{proof}

\begin{cor}
\label{cor-Q-int}
The Baxter operator acts on $\Fc_n$ by the following integral kernel:
$$
Q_n(u) f(\bs x) = \zeta^{-n} \int e^{2\pi i(c_b-u)\hr{\underline{\bs x} -\underline{\bs y}}} \prod_{k=1}^{n}\frac{1}{\varphi(y_k-x_k-c_b)}\prod_{k=1}^{n-1}\frac{\varphi(y_{k+1}-y_k)}{\varphi(x_{k+1}-y_k)} f(\bs y) d\bs y,
$$
where we set $\bs x \in \R^n$ and $\bs y \in \C^n$. Note that in agreement with Notation~\ref{contour-convention} the contour of integration can be chosen to be $(\R+i0)^n$.
\end{cor}

\begin{proof}
The statement follows from consecutive application of Lemma~\ref{Tf} and the following formula
$$
\varphi(p_1+u) f(\bs x) = \zeta^{-1} \int \frac{e^{2\pi i(x_1-y_1)(c_b-u)}}{\varphi(y_1-x_1-c_b)} f(y_1, x_2, \dots, x_n) dy_1,
$$
whose proof is analogous to that of Lemma~\ref{Tf}
\end{proof}

\begin{cor}
\label{cor-Q-inv}
The inverse of the Baxter operator acts on $\Fc_n$ by
$$
Q_n^{-1}(u) f(\bs x) = \zeta^n \int e^{2\pi i(c_b+u)\hr{\underline{\bs z} -\underline{\bs x}}} \prod_{k=1}^{n} \varphi(x_k-z_k+c_b) \prod_{k=1}^{n-1}\frac{\varphi(z_{k+1}-x_k)}{\varphi(x_{k+1}-x_k)} f(\bs z) d\bs z.
$$
\end{cor}

\subsection{Operator identities}

We conclude this section with a pair of operator identities that prove useful in the sequel.

\begin{prop}
\label{prop-phi-inv}
If the function $f(\bs x)$ is such that $p_{n+1} f(\bs x)=0$, then
\begin{align*}
\varphi(p_{n+1}+x_{n+1}+\alpha+c_b) f(\bs x) = \varphi(x_{n+1}+\alpha)^{-1} f(\bs x), \\
\varphi(x_{n+1}-p_{n+1}+\alpha-c_b)^{-1} f(\bs x) = \varphi(x_{n+1}+\alpha) f(\bs x).
\end{align*}
\end{prop}

\begin{proof}
We shall only prove the first equality, the proof of the second one being completely analogous. First, note that the Fourier transform~\eqref{Fourier-1} yields
$$
\varphi(p_{n+1}+x_{n+1}+\alpha+c_b) f(\bs x)  = \zeta^{-1} \int \frac{e^{2\pi i t(p_{n+1}+x_{n+1}+\alpha)}}{\varphi(t-c_b)} f(\bs x) dt.
$$
Now, we can rewrite the right hand side as
$$
\zeta^{-1} \int \frac{e^{\pi i t^2} e^{2\pi i t(x_{n+1}+\alpha)} e^{2\pi i t p_{n+1}}}{\varphi(t-c_b)} f(\bs x) dt = \zeta^{-1} \int \frac{e^{\pi i t^2} e^{2\pi i t(x_{n+1}+\alpha)}}{\varphi(t-c_b)} f(\bs x) dt
$$
where we use that $f(\bs x)$ is annihilated by $p_{n+1}$. Changing the sign of the integration variable, using the inversion formula for the quantum dilogarithm, and the Fourier transform~\eqref{Fourier-2}, we now see that the right hand side takes form
$$
\zeta \int e^{-2\pi i t(x_{n+1}+\alpha+c_b)} \varphi(t+c_b) f(\bs x) dt = \varphi(x_{n+1}+\alpha)^{-1} f(\bs x).
$$
The latter equality follows from relation~\eqref{Fourier-2}.
\end{proof}

\begin{cor}
\label{cancel-cor}
If the function $f(x)$ is such that $p_{n+1} f(x)=0$, then
$$
S_{n+1}(u+c_b) f(x) = \breve S_{n+1}(u)^{-1} f(x),
$$
where
$$
\breve S_{n+1}(u) = \varphi(p_n+x_{n+1}-x_{n-1}+u) \varphi(p_n+x_{n+1}-x_{n}+u).
$$
\end{cor}

\begin{proof}
The proof consists of applying Proposition~\ref{prop-phi-inv} twice.
\end{proof}

\section{Quantum Toda Hamiltonians}
\label{sect-Ham}

We now explain how the Hamiltonians of the $q$-deformed open Toda system can be recovered from the Baxter operator. To this end, we consider the operator defined on the space of test functions $\Fc_n$ by
\beq
\label{A-series}
A_n(u) = Q_n(u-ib/2)Q_n(u+ib/2)^{-1}.
\eeq

\begin{prop}
\label{prop-A}
We have
$$
A_{n+1}(u) = \hr{1 + e^{2\pi b(p_{n+1}+u)}} A_n(u) + e^{2\pi b(p_{n+1}+x_{n+1}-x_n+u)} A_{n-1}(u),
$$
\end{prop}

\begin{proof}
Using the functional equation~\eqref{eqn-func} along with the equality
$$
\varphi(x-ib/2) e^{2\pi bp} = e^{2\pi bp} \varphi(x+ib/2)
$$
one immediately derives relations
\beq
\label{fracT-1}
\frac{T_{n+1}(u-ib/2)}{T_{n+1}(u+ib/2)} = 1 + e^{2\pi b(p_{n+1}+u)} + e^{2\pi b(p_{n+1}+x_{n+1}-x_n+u)}
\eeq
and
\beq
\label{fracT-2}
T_n(u-ib/2) e^{2\pi b(p_{n+1}+x_{n+1}-x_n+u)} = e^{2\pi b(p_{n+1}+x_{n+1}-x_n+u)} T_n(u+ib/2).
\eeq
The result then follows from the definition~\eqref{A-series}, expanding
$$
Q_{n+1}(u) = Q_{n-1}(u) T_n(u) T_{n+1}(u),
$$
and using equalities~\eqref{fracT-1} and~\eqref{fracT-2}.
\end{proof}

\begin{defn}
\label{defn-A}
By Proposition~\ref{prop-A} we can expand
$$
A_n(u) = \sum_{k=0}^n H_k^{(n)} e^{2\pi bku}.
$$
We define the \emph{$k$-th $q$-deformed $\gl_n$ open Toda Hamiltonian} to be the operator $H_k^{(n)}$.
\end{defn}

From Lemma~\ref{lem-QQ}, we immediately deduce
\begin{cor}
The Toda Hamiltonians $H_0^{(n)},\ldots, H_{n}^{(n)}$ form a commuting set of operators. 
\end{cor}

Proposition~\ref{prop-A} is equivalent to the following recursive description of the Toda Hamiltonians:

\begin{cor}
\label{cor-Ham}
The following formula holds:
$$
H_k^{(n+1)} = H_k^{(n)} + e^{2 \pi b p_{n+1}} H_{k-1}^{(n)} + e^{2 \pi b (p_{n+1} + x_{n+1} -x_n)} H_{k-1}^{(n-1)}.
$$
\end{cor}

\begin{example}
Since
$$
A_0(u) = 1 \qquad\text{and}\qquad A_1(u) = \frac{\varphi(p_1+u-ib/2)}{\varphi(p_1+u+ib/2)} = 1+e^{2\pi bu}e^{2\pi bp_1},
$$
Definition~\ref{defn-A} yields
$$
H_0^{(0)} = H_0^{(1)} = 1 \qquad\text{and}\qquad H_1^{(1)} = e^{2\pi bp_1}.
$$
From Corollary~\ref{cor-Ham}, we find the first $\gl_2$ Toda Hamiltonian to be
$$
H_1^{(2)} = e^{2\pi bp_1} + e^{2\pi bp_2} + e^{2\pi b(p_2+x_2-x_1)}.
$$
\end{example}

\begin{remark}
In a completely analogous fashion, one can define a set of commuting Hamiltonians $\widetilde{H}_k^{(n)}$ associated to the bottom Baxter operator by means of the generating series
$$
\widetilde{A}_n(u) = Q^{\mathrm{b}}_n(u+ib/2)Q^{\mathrm{b}}_n(u-ib/2)^{-1} = \sum_{k=0}^n \widetilde{H}_k^{(n)} e^{-2\pi bku}.
$$
In view of Proposition~\ref{tb-commut}, these Hamiltonians also commute with the Hamiltonians $H_k^{(n)}$: for all $j,k$ one has 
$$
\big[\widetilde{H}_k^{(n)},H_j^{(n)}\big]=0.
$$
Moreover, using an analogue of Corollary~\ref{cor-Ham} for the Hamiltonians $\widetilde{H}_k^{(n)}$, it is straightforward to derive the following relation inductively:
\beq
\label{old-new-Ham}
H^{(n)}_n \widetilde{H}^{(n)}_k = H^{(n)}_{n-k}.
\eeq
\end{remark}

\section{The $b$\,-Whittaker functions}

In this section we use the Baxter operators to define the $b$\,-Whittaker functions by means of a Givental-type integral formula, and establish some of their basic properties. A similar approach has been taken in \cite{GKLO08} in the case of undeformed Whittaker functions, and in \cite{GKLO14} for Macdonald polynomials and their degenerations, the $q$-Whittaker polynomials. We also explain the relation of our Givental formula to the Mellin--Barnes integral representation of the $b$\,-Whittaker functions that was first discovered in \cite{KLS02}. 

\subsection{A modular $b$\,-analog of Givental's formula}
\label{givental-subsec}
Let us introduce an operator $Q_n^{n+1}(u)$ defined by
\beq
\label{rec-op}
Q_n^{n+1}(u) = Q_n(u)P_{n+1}(u),
\eeq
where $Q_n(u)$ is the $\gl_n$ Baxter $Q$-operator, and
$$
P_{n+1}(u) = \frac{e^{2\pi i (c_b-u)x_{n+1}}}{\varphi(x_{n+1}-x_n)}.
$$
If $f \colon \C^n \to \C$ is a function of $\bs x \in \C^n$, then applying the operator $Q_n^{(n+1)}(u)$ to~$f$ produces a new function $\check f = Q_n^{(n+1)}(u) f$ that now depends on variables $(\bs x, x_{n+1})$. Throughout the text, we conventionally suppress the dependence of the operator $Q_n^{(n+1)}(u)$ on the added variable $x_{n+1}$.

\begin{cor}
\label{rec-act}
The action of the operator $Q_n^{n+1}(u)$ on a test function $f\in\mathcal{F}_n$ is given by
$$
Q^{n+1}_n(u) f(\bs x) = \zeta^{-n} \int \frac{e^{2\pi i(c_b-u)\hr{\underline{\bs x}  -\underline{\bs y}}} \prod_{k=1}^{n-1}\varphi(y_{k+1}-y_k)}{\prod_{k=1}^n \varphi(y_k-x_k-c_b)\varphi(x_{k+1}-y_k)} f(\bs y) d\bs y
$$
where $\bs x \in \R^{n+1}$ and $\bs y \in \C^n$.
\end{cor}

\begin{proof}
Follows from Corollary~\ref{cor-Q-int}.
\end{proof}

We now use the operators $Q^{n+1}_n(u)$ to give a recursive definition of the $b$\,-Whittaker functions.

\begin{defn}
\label{def-whit}
We define the $b$\,-Whittaker function for $\gl_1$ with spectral variable $\la$ to be
$$
\Psi^{(1)}_{\lambda}(x)  =  e^{2\pi i \lambda x}.
$$
The $b$\,-Whittaker function for $\gl_n$ with spectral variables $(\bs\la,\la_{n+1})$ is then defined inductively to be
$$
\Psi^{(n+1)}_{\bs\la,\lambda_{n+1}}(\bs x, x_{n+1}) = e^{\pi i c_b \sum_{j=1}^n (\la_{n+1}-\la_j)} Q^{n+1}_n(c_b - \lambda_{n+1}) \Psi^{(n)}_{\bs\la}(\bs x).
$$
\end{defn}

Combining Definition~\ref{def-whit} with the explicit form of the kernel for $Q_n^{n+1}(u)$ given in Corollary~\ref{rec-act}, we obtain the following explicit integral formula for the $b$\,-Whittaker functions.

\begin{prop}
\label{giv-prop}
The $b$\,-Whittaker function for $\gl_n$ can be written as follows:
\beq
\label{G-closed}
\begin{aligned}
\Psi^{(n)}_{\bs\la}(\bs x) = \zeta^{-\frac{n(n-1)}{2}} e^{2\pi i(\la_n\underline{\bs x}-c_b\rho_n(\bs\la))} \int \prod_{j=1}^{n-1} \frac{e^{2\pi i\underline{\bs t}_j(\la_j-\la_{j+1})} \prod_{k=2}^j \varphi(t_{j,k} - t_{j,k-1}) \, d\bs t_j}{\prod_{k=1}^j \varphi(t_{j,k} - t_{j+1,k} - c_b) \varphi(t_{j+1,k+1} - t_{j,k})},
\end{aligned}
\eeq
where
$$
\bs t_j = (t_{j,1}, \dots, t_{j,j})
\qquad\text{and}\qquad
\bs t_n = \bs x.
$$
\end{prop}

This formula is a modular $b$\,-analog of Givental's formula for the undeformed $\gl_n$ Whittaker functions as an integral over triangular arrays, see \cite{Giv97},\cite{GKLO06}. It can be established by an inductive argument that the integral in~\eqref{G-closed} is correctly defined and converges absolutely for real $\bs \la,\bs x$. Let us treat the base case $n=2$ first. For any $0 < \eps < \Im(c_b)$, we can write
\begin{align*}
\Psi^{(2)}_{\la_1,\la_2}(x_1, x_2)
&= \zeta^{-1} e^{\pi i c_b (\la_2-\la_1)} \int_{\R + i\eps} \frac{e^{2\pi i \la_2 (\underline{\bs x} - t)}} {\varphi(t-x_1-c_b)\varphi(x_2-t)} e^{2\pi i \la_1 t} dt \\
&= \zeta^{-1} e^{\pi (ic_b+2\eps) (\la_2-\la_1)} e^{2\pi i\la_2 \underline{\bs x}} \int_\R \frac{e^{2\pi i t (\la_1 - \la_2)}} {\varphi(t+i\eps-x_1-c_b)\varphi(x_2-t-i\eps)} dt.
\end{align*}
 Absolute convergence of the latter integral follows easily from the asymptotic behavior of the function $\varphi(z)$. Indeed, as $t \to +\infty$ the absolute value of the integrand behaves as
$$
\frac1{\hm{\varphi(t+i\eps-x_1-c_b)\varphi(x_2-t-i\eps)}} \sim \hm{e^{-\pi i(t+i\eps-x_1-c_b)^2}} = e^{\pi(x_1-t)(b+b^{-1}-2\eps)}.
$$
Similarly, as $t \to -\infty$, we get
$$
\frac{1}{\hm{\varphi(t+i\eps-x_1-c_b)\varphi(x_2-t-i\eps)}} \sim \hm{e^{-\pi i(x_2-t-i\eps)^2}} = e^{2\pi\eps(t-x_2)},
$$
Both of the above expressions are of exponential decay for any $0 < \eps < \Im(c_b)$, and hence the integral converges absolutely.

The inductive argument now proceeds as follows. Having defined $b$\,-Whittaker functions of rank $n$, we establish their orthogonality and completeness (Section~\ref{thm-proof}) and thereby derive for them a Mellin--Barnes representation~\eqref{MB-rec}. This Mellin--Barnes formula is particularly well-suited to estimating the coordinate asymptotics of the $b$\,-Whittaker functions, see Proposition~\ref{prop-asymp}. With these asymptotics in hand, absolute convergence of the higher rank integrals is then easily established as in the rank 2 case.

\subsection{Symmetry in spectral variables}
Although not immediately obvious from the integral representation~\eqref{G-closed}, the $b$\,-Whittaker function $\Psi^{(n)}_{\bs\la}(\bs x)$ is a symmetric function of its spectral variables $\la_1,\ldots,\la_n$.
\begin{lemma}
\label{permut-lem}
We have
$$
e^{\pi ic_b(v-u)} T_n(u)P_{n+1}(u)P_{n}(v) = e^{\pi ic_b(u-v)} T_n(v)P_{n+1}(v)P_{n}(u)
$$
\end{lemma}
\begin{proof}
Using Lemma~\ref{Tf}, we express
$$
T_n(u)P_{n+1}(u)P_{n}(v) = \frac{e^{2\pi i(c_b-u)(x_n+x_{n+1})}}{\varphi(x_n-x_{n-1})} \int \frac{e^{2\pi i y(u-v)}}{\varphi(y-x_n-c_b)\varphi(x_{n+1}-y)}dy.
$$
Shifting the integration variable $y \to y+x_n$ we get
$$
T_n(u)P_{n+1}(u)P_{n}(v) = \frac{e^{-2\pi i (\tilde vx_n+\tilde ux_{n+1})}}{\varphi(x_n-x_{n-1})} \int \frac{e^{2\pi i y(u-v)}}{\varphi(y-c_b)\varphi(x_{n+1}-x_n-y)}dy.
$$
where
$$
\tilde u = u-c_b \qquad\text{and}\qquad \tilde v = v-c_b.
$$
We now use Formula~\eqref{Fourier-2} to derive
$$
\int \frac{e^{2\pi i y(u-v)}}{\varphi(y-c_b)\varphi(x_{n+1}-x_n-y)}dy= \zeta \int \frac{\varphi(t+c_b)}{\varphi(y-c_b)} e^{2\pi i y(u-v)} e^{2\pi i t(y+x_n-x_{n+1}-c_b)} dt dy
$$
and Formula~\eqref{Fourier-1} to take an integral over $y$. This way we obtain
$$
\int \frac{e^{2\pi i y(u-v)}}{\varphi(y-c_b)\varphi(x_{n+1}-x_n-y)}dy= \zeta^2 \int e^{2\pi i t(x_n-x_{n+1}-c_b)} \varphi(t+u-v+c_b) \varphi(t+c_b) dt.
$$
Collecting everything together, and shifting the integration variable $t \to t+v$, we see that the left hand side in the statement of the Lemma equals
$$
\zeta^2 e^{-\pi i (u+v) (2x_{n+1}+c_b)} \frac{e^{2\pi i c_b (x_n+x_{n+1})}}{\varphi(x_n-x_{n-1})} \int e^{2\pi i t(x_n-x_{n+1}-c_b)} \varphi(t+u+c_b) \varphi(t+v+c_b) dt.
$$
Now, the result follows from the symmetry of the latter expression in variables $u$ and $v$.
\end{proof}

\begin{prop}
The $b$\,-Whittaker function $\Psi^{(n)}_{\bs\la}(\bs x)$ is a symmetric function in $\bs\la$. 
\end{prop}

\begin{proof}
It suffices to check invariance under the simple reflections $(\la_k,\la_{k+1})\mapsto(\la_{k+1},\la_k)$. This invariance amounts to establishing the identity 
\beq
\label{Qsym}
e^{\pi i c_b(v-u)} Q^{n+1}_n(u)Q^{n}_{n-1}(v) = e^{\pi i c_b(u-v)} Q^{n+1}_n(v)Q^{n}_{n-1}(u),
\eeq
where we understand $Q^1_0(u) = e^{2\pi i(c_b-u)x_1}$. To this end, note that  for $n\geq 2$ we have
\begin{align*}
Q^{n+1}_n(u)Q^{n}_{n-1}(v)
&= Q_n(u)P_{n+1}(u)Q_{n-1}(v)P_{n}(v) \\
&= Q_{n-1}(u)T_n(u)Q_{n-1}(v)P_{n+1}(u)P_{n}(v) \\
&= Q_{n-1}(u)Q_{n-1}(v)S_n(u+v)T_n(u)P_{n+1}(u)P_{n}(v).
\end{align*}
Applying Lemmas~\ref{lem-QQ} and~\ref{permut-lem} we see that
\begin{align*}
e^{\pi i c_b(v-u)} Q^{n+1}_n(u)Q^{n}_{n-1}(v)
&= e^{\pi i c_b(v-u)} Q_{n-1}(u)Q_{n-1}(v)S_n(u+v)T_n(u)P_{n+1}(u)P_{n}(v) \\
&= e^{\pi i c_b(u-v)} Q_{n-1}(v)Q_{n-1}(u)S_n(u+v)T_n(v)P_{n+1}(v)P_{n}(u),
\end{align*}
Since the last two lines can be obtained from one another by exchanging $u$ and $v$, we see that identity~\eqref{Qsym} holds. By a similar but simpler calculation, one verifies
$$
e^{\pi i c_b(v-u)} Q^2_1(u)Q^1_0(v) = e^{\pi i c_b(u-v)} Q^2_1(v)Q^1_0(u),
$$
and the result follows.  
\end{proof}

\subsection{Relation with the Mellin--Barnes representation}
Recall that in Definition~\ref{def-whit}, the $b$\,-Whittaker function for $\gl_{n+1}$ is obtained from that the $\gl_n$ $b$\,-Whittaker function by applying an integral operator acting in the coordinate variables. However, by the completeness relation for the $\mathfrak{gl}_n$ $b$\,-Whittaker functions in Theorem~\ref{completeness}, we have an expansion

$$
\Psi^{(n+1)}_{\bs\la}(\bs x, x_{n+1}) = \int K_n(\bs\mu, \bs\la ; x_{n+1}) \Psi^{(n)}_{\bs\mu}(\bs x) m(\bs\mu) d\bs\mu.
$$
The orthogonality relation for the $b$\,-Whittaker functions of rank $n$ in Theorem~\ref{orthogonality} allows us to write
$$
K_n(\bs\mu, \bs\la ; x_{n+1}) = \int \Psi^{(n+1)}_{\bs{\la'},\la_{n+1}}(\bs x, x_{n+1}) \overline{\Psi_\mu(\bs x)} d\bs x.
$$
By Definition~\ref{def-whit} we have
\begin{align*}
e^{\pi ic_b\sum_{j=1}^{n}(\la_j - \la_{n+1})} K_n(\bs\mu, \bs\la ; x_{n+1})
&= \int Q_n^{(n+1)}(c_b-\la_{n+1}) \Psi^{(n)}_{\bs{\la'}}(\bs x) \, \overline{\Psi_\mu(\bs x)} d \bs x \\
&= \int P_{n+1}(c_b-\la_{n+1}) \Psi^{(n)}_{\bs{\la'}}(\bs x) \,\overline{Q_n(-c_b-\la_{n+1})^{-1} \Psi_\mu(\bs x)} d \bs x.
\end{align*}
Using Proposition~\ref{Q-eigen-prop} we see that the latter integral is equal to the product 
$$
C_n(\bs{\la'}, \bs\mu, x_{n+1}) e^{2\pi ix_{n+1}\la_{n+1}} \prod_{j=1}^n \varphi(\mu_j-\la_{n+1}+c_b),
$$
where $C_n$ is the integral defined in~\eqref{C-int}. Finally, applying Theorem~\ref{cauchy-littlewood}, we arrive at the equality
$$
K_n(\bs\mu, \bs\la ; x_{n+1}) = L_n(\bs\mu, \bs\la ; x_{n+1}) e^{\frac{\pi in}{2}\sum_{j=1}^{n+1}\la_j^2} e^{-\frac{\pi i(n-1)}{2}\sum_{j=1}^{n}\mu_j^2},
$$
with
$$
L_n(\bs\mu, \bs\la ; x_{n+1}) =  \zeta^{-n} e^{\pi i(2x_{n+1}-\underline{\bs\mu})(\underline{\bs\la}-\underline{\bs\mu})} \prod_{j=1}^{n+1} \prod_{k=1}^n c(\la_j-\mu_k).
$$

Thus, if we define the \emph{Mellin-Barnes normalization} of the $b$\,-Whittaker function to be
$$
\psi^{(n)}_{\bs\la}(\bs x) = e^{\frac{\pi i(1-n)}{2}\sum_{j=1}^{n}\la_j^2} \Psi^{(n)}_{\bs\la}(\bs x),
$$
then the following recursive formula holds:	
\beq
\label{MB-rec}
\psi^{(n+1)}_{\bs\la}(\bs x,x_{n+1}) = \int L_n(\bs\mu, \bs\la ; x_{n+1}) \psi^{(n)}_{\bs\la}(\bs x) m(\bs\mu) d\bs\mu
\eeq
We refer to formula~\eqref{MB-rec} as the \emph{Mellin--Barnes} formula for the $b$\,-Whittaker functions.  It expresses the $b$\,-Whittaker function for $\gl_{n+1}$ as the result of applying an integral operator acting on the spectral variables of the $b$\,-Whittaker function for $\gl_n$. 

\begin{remark}
The Mellin--Barnes formula for the $b$\,-Whittaker functions was first derived in~\cite[Theorem 3.1]{KLS02}. We caution the reader that in {\it loc. cit.} the recursion is performed with respect to the last $n$ coordinate variables $(x_2,\ldots, x_{n+1})$, rather than over the first $n$ as we do in~\eqref{MB-rec}.
\end{remark}

\subsection{Analytic properties of the $b$\,-Whittaker functions}
Here we collect some basic analytic properties of the $b$\,-Whittaker functions.
\begin{prop}
The $b$\,-Whittaker function $\Psi^{(n)}_{\bs\la}(\bs x)$ can be analytically continued to an entire function of $\bs\la$.  As a function of $\bs x$ it can be analytically continued to a meromorphic function, such that the product $\prod_{k=1}^{n-1}\varphi(x_{k+1}-x_k)\Psi^{(n)}_{\bs\la}(\bs x)$ is entire. 
\end{prop}

\begin{proof}
The analytic continuation in $\bs x$ is easily established by deformation of the contours in the Givental formula~\eqref{G-closed} in accordance with Convention~\ref{analytic-continuation}. The continuation in $\bs\la$ is constructed similarly with the help of the Mellin--Barnes formula~\eqref{MB-rec}. 
\end{proof}

The following lemma describes bounds on the asymptotic behavior of the $b$\,-Whittaker functions with respect to coordinate and spectral variables. 

\begin{prop}
\label{prop-asymp}
The  $b$\,-Whittaker functions have the following asymptotic properties:
\begin{enumerate}
\item As a function of $\bs\la$, $\Psi^{(n)}_{\bs\la}(\bs x)$ is rapidly decaying along any direction such that for some $j \ne k$ we have $\la_j - \la_k \to \infty$. In the regime $\la_j = \la+a_j$, where $a_j$ is constant for all $j=1, \dots, n$ and $\la \to \infty$, we have $\Psi^{(n)}_{\bs\la}(\bs x) =  C(\bs x, \bs a)e^{2 \pi i\la\underline{\bs x}}$ where $C(\bs x,\bs a)$ is independent of $\bs \la$, and $\bs a = (a_1, \dots, a_n)$.
\item As a function of $\bs x$, $\Psi^{(n)}_{\bs\la}(\bs x)$ is rapidly decaying as $\bs x$ tends to infinity along any line outside of the chamber $x_1 \ge x_2 \ge \dots \ge x_n$. Its asymptotic behavior in the latter chamber is given by
\beq
\label{x-residues}
\Psi^{(n)}_{\bs\la}(\bs x) \sim \sum_{w\in S_n} e^{2\pi i \sum_{j=1}^n x_j\la_{w(j)}} \prod_{j<k} e^{\frac{\pi i }{2} (\la_j - \la_k)^2} c\hr{\la_{w(k)} - \la_{w(j)}}.
\eeq
where $S_n$ is the permutation group on $n$ elements.
\end{enumerate}
\end{prop}

\begin{proof}
	First we establish the asymptotics in $\bs\la$ with the help of the Givental formula~\eqref{G-closed}.
To illustrate the argument, recall that in the $\gl_2$ case we have
\beq
\label{gl2-giv}
\Psi^{(2)}_{\la_1,\la_2}(x_1, x_2) = \zeta^{-1} e^{\pi (ic_b+2\eps) (\la_2-\la_1)} e^{2\pi i\la_2 \underline{\bs x}} \int_\R \frac{e^{2\pi i t (\la_1 - \la_2)}} {\varphi(t+i\eps-x_1-c_b)\varphi(x_2-t-i\eps)} dt,
\eeq
where $0 < \eps < \Im(c_b)$. As explained in Section~\ref{givental-subsec}, the integral in the latter expression converges absolutely. Therefore, for any $0 < \eps < \frac{b+b^{-1}}{2}$ the $b$\,-Whittaker function for $\gl_2$ satisfies
$$
\left| \Psi^{(2)}_{\la_1,\la_2}(x_1, x_2)\right| \leq F(x_1,x_2)e^{\pi (ic_b+2\eps) (\la_2-\la_1)},
$$
where $F(x_1,x_2)$ is independent of $\bs\la$, which establishes the rapid decay if $|\la_1-\la_2|\rightarrow\infty$. And in the limit $\la_1 = \la_2 +c$ for constant $c$,  the only $\la$-dependence in formula~\eqref{gl2-giv} arises from the factor $e^{2\pi i\la_2\underline{\bs x}}$. The case of $b$\,-Whittaker functions of higher rank is treated similarly, using the formula~\eqref{G-closed}.

The coordinate asymptotics described in part (2) of the Proposition are similarly determined with the help of the Mellin--Barnes formula~\eqref{MB-rec}. The proof of the asymptotics outside the chamber $x_1 \ge x_2 \ge \dots \ge x_n$ is identical to that of point (1), shifting the integration contours downwards in the Mellin-Barnes formula~\eqref{MB-rec} for the $b$\,-Whittaker function. The asymptotics inside the positive chamber follows from Cauchy's theorem: shifting the contour of integration in~\eqref{MB-rec} upwards to cross the poles of the integrand on the real axis, we can express the $b$\,-Whittaker function as the sum of residues in the expression~\eqref{x-residues} along with the contribution from the integral along the shifted contour, the latter being rapidly decaying by the usual argument. The Proposition is proved.
\end{proof}

\begin{remark}
It follows from Proposition~\ref{prop-asymp} that the $b$\,-Whittaker function $\Psi^{(n)}_{\bs\la}(\bs x)$ defines a tempered distribution on the (classical) Schwartz space of functions in $x_1,\ldots, x_n$. 
\end{remark}

\section{$b$\,-Whittaker functions as eigenvectors}

We now show that the $b$\,-Whittaker functions are distributional eigenvectors of the Baxter operator, Toda Hamiltonians, and the Dehn twist operator.
 
\begin{lemma}
\label{aux-lem}
We have
\begin{align*}
&T_n(u)P_{n+1}(v) = P_{n+1}(v) T_n(u) \breve S_{n+1}(u), \\
&T_n(u)P_n(v) = P_n(v)\varphi(p_n+u-v+c_b), \\
&S_n(u)P_n(v) = P_n(v)S_n(u-v+c_b).
\end{align*}
\end{lemma}

\begin{proof}
All statements are proven by direct computations using the pentagon equation and Lemma~\ref{triv}.
\end{proof}

\begin{lemma}
\label{eigen-induction}
Let $f(\bs x)$ be a function annihilated by $p_{n+1}$. Then the following equalities holds
$$
Q_{n+1}(u) Q^{n+1}_n(v) f(\bs x) = \varphi(u-v+c_b) Q^{n+1}_{n}(v) Q_{n}(u) f(\bs x),
$$
$$
Q^{\mathrm{b}}_{n+1}(u) Q^{n+1}_n(v) f(\bs x) = \varphi(v-u-c_b) Q^{n+1}_{n}(v) Q^{\mathrm{b}}_{n}(u) f(\bs x),
$$
\end{lemma}

\begin{proof}
We establish the identity for the top Baxter operator $Q_{n+1}(u)$, then the proof for $Q^{\mathrm{b}}_{n+1}(u)$ follows from relation~\eqref{old-new-Ham}. First, using Lemmas~\ref{lem-QQ} and~\ref{lem-TT} we arrive at
$$
Q_{n+1}(u) Q^{n+1}_n(v) = Q_{n}(v) Q_{n}(u) S_{n+1}(u+v) T_{n+1}(u) P_{n+1}(v).
$$
Then, by Lemma~\ref{aux-lem} and condition $p_{n+1} f(\bs x) = 0$ we get
\begin{align*}
Q_{n+1}(u) Q^{n+1}_n(v)
&= Q_{n}(v) Q_{n}(u) S_{n+1}(u+v) P_{n+1}(v) \varphi(p_{n+1}+u-v+c_b) \\
&= Q_{n}(v) Q_{n}(u) P_{n+1}(v) S_{n+1}(u+c_b) \varphi(u-v+c_b).
\end{align*}
Finally, apply Corollary~\ref{cancel-cor} and Lemma~\ref{aux-lem} once again we obtain
\begin{align*}
Q_{n+1}(u) Q^{n+1}_n(v)
&= \varphi(u-v+c_b) Q_{n}(v) Q_{n-1}(u) T_{n}(u) P_{n+1}(v) S_{n+1}(u+c_b) \\
&= \varphi(u-v+c_b) Q^{n+1}_{n}(v) Q_{n}(u) \breve S_{n+1}(u) S_{n+1}(u+c_b) \\
&= \varphi(u-v+c_b) Q^{n+1}_{n}(v) Q_{n}(u),
\end{align*}
which concludes the proof.
\end{proof}
Iteratively applying Lemma~\ref{eigen-induction}, we obtain

\begin{prop}
\label{Q-eigen-prop}
The $b$\,-Whittaker functions are eigenvectors of the Baxter operators: we have
\begin{align}
\label{eig-top}
Q^{\mathrm{t}}_n(u) \Psi^{(n)}_{\bs\la}(\bs x) &= \prod_{j=1}^n \varphi(u+\la_j) \Psi^{(n)}_{\bs\la}(\bs x), \\
\label{eig-bottom}
Q^{\mathrm{b}}_n(v) \Psi^{(n)}_{\bs\la}(\bs x) &= \prod_{j=1}^n \varphi(-v-\la_j) \Psi^{(n)}_{\bs\la}(\bs x).
\end{align}
\end{prop}
Using the functional equation~\eqref{eqn-func}, we deduce

\begin{prop}
\label{toda-eigen-prop}
The $b$\,-Whittaker functions are common eigenvectors of the Toda Hamiltonians: we have
\begin{align*}
H_k^{(n)} \Psi^{(n)}_{\bs\la}(\bs x) &= e_k(\bs \la) \Psi^{(n)}_{\bs\la}(\bs x), \\
\widetilde{H}_k^{(n)} \Psi^{(n)}_{\bs\la}(\bs x) &= e_k(-\bs \la) \Psi^{(n)}_{\bs\la}(\bs x)
\end{align*}
where $e_k(\bs \la)$ is the $k$-th elementary symmetric function in variables $e^{2\pi b\la_j}$ for $j=1, \dots, n$.
\end{prop}

\begin{prop}
\label{dehn-eigen-prop}
The $b$\,-Whittaker function $\Psi^{(n)}_{\bs\la}(\bs x)$ is an eigenfunction of the Dehn twist operator $D_n$ with eigenvalue given by
$$
D_n \Psi^{(n)}_{\bs\la}(\bs x) = e^{-\pi i \sum_{j=1}^n\la_j^2}\Psi^{(n)}_{\bs\la}(\bs x).
$$
\end{prop}

\begin{proof}
Recall from Proposition~\ref{degen-dehn} that
$$
D_n = \zeta_{\inv}^n~Q_{n}^{\mathrm{swap}}(0,0)^{-1}.
$$
The statement therefore follows by setting $u=v=0$ in formulas~\eqref{eig-top} and \eqref{eig-bottom} and applying the inversion formula~\eqref{inv}. 
\end{proof}

\section{Integral identities for the $b$\,-Whittaker functions }

In this section we study integrals of the form
\beq
\label{matrix-coeff}
\int_{\R^n} f(x_n) \Psi^{(n)}_{\bs\la}(\bs x) \overline{\Psi^{(n)}_{\bs{\bar\mu}}(\bs x) } d\bs x,
\eeq
where $f$ is a meromorphic function of a single variable. Our two main cases of interest are when $f(x_n) = e^{2\pi izx_n}$ or $f(x_n) = \varphi(u-x_n)^{-1}$, which we treat in Sections~\ref{subsec-ort-int} and~\ref{subsec-CL} respectively. The evaluation of the former integral is the main step in the proof of orthogonality of $b$\,-Whittaker functions. The latter integral is perhaps even more interesting: it can be regarded as a modular $b$\,-analog of the Cauchy-Littlewood identity. More specifically, it is a $b$\,-deformation of the Whittaker analog~\cite[(3.21)]{COSZ14} of the Cauchy-Littlewood formula for Schur functions, which was first proved by Stade in~\cite{Sta02}, see also~\cite{Lam13}. The identity in ~\cite{Sta02,COSZ14} involves $\Gamma$-functions in place of $c$-functions, and ordinary  undeformed Whittaker functions instead of their modular $b$\,-deformed counterparts. In the context of the {\it XXX} spin chain~\cite{DKM01, DKM03, DM14} and the undeformed Toda system~\cite{Sil07}, integrals analogous to~\eqref{matrix-coeff} have been previously studied using the Feynman diagram technique.

\begin{notation}
In the sequel we shall use the notation
$$
\ha{f(\bs x),g(\bs x)} = \int_{\R^n} f(\bs x) \overline{g(\bs x)} d\bs x.
$$
\end{notation}

\subsection{Matrix coefficients}

In what follows it will be helpful to consider operators
$$
Y_n(u) = \varphi(p_n-x_n+u)^{-1}\varphi(p_n-x_{n-1}+u)^{-1}.
$$

\begin{lemma}
\label{y-dehn}
If $D_n$ is the Dehn twist operator, we have
$$
Y_n(u)D_n=D_n\varphi(u-x_{n})^{-1}.
$$
\end{lemma}

\begin{proof}
By the pengaton identity, we have
$$
Y_{n}(u) = \varphi(x_{n}-x_{n-1})^{-1}\varphi(p_{n}-x_{n}+u)^{-1} \varphi(x_{n}-x_{n-1}).
$$
Therefore we can write
$$
Y_{n}(u)D_n = \varphi(x_{n}-x_{n-1})^{-1} \varphi(p_{n}-x_{n}+u)^{-1} \prod_{j=1}^{n-2} \varphi(x_{j+1}-x_j)^{-1} \prod_{j=1}^n e^{-\pi i p_j^2}.
$$
Commuting the second factor in the above expression all the way to the right we obtain
$$
Y_n(u)D_n=D_n\varphi(u-x_{n})^{-1}.
$$
\end{proof}

\begin{lemma}
\label{conj-Q}
We have
$$
Q_n^{(n+1)}(v)^* Q_n^{(n+1)}(u) = e^{2\pi i (2c_b+\overline{v}-u)x_{n+1}} Q_n(u) Y_n(x_{n+1}+u)^{-1}  Y_n(x_{n+1}+\overline{v}) Q_n(\overline v)^{-1}.
$$
\end{lemma}

\begin{proof}
First, we use formula~\eqref{rec-op} and Lemma~\ref{lem-QQ} to write
\begin{align*}
Q_n^{(n+1)}(v)^* Q_n^{(n+1)}(u)
&= e^{2\pi i (2c_b+\overline{v}-u)x_{n+1}} \varphi(x_{n+1}-x_n) Q_n(\bar v)^{-1} Q_n(u) \varphi(x_{n+1}-x_n)^{-1} \\
&= e^{2\pi i (2c_b+\overline{v}-u)x_{n+1}} \varphi(x_{n+1}-x_n) Q_n(u) Q_n(\bar v)^{-1} \varphi(x_{n+1}-x_n)^{-1}.
\end{align*}
The rest of the proof follows from equality
$$
\varphi(x_{n+1}-x_n) Q_n(u) = Q_n(u) Y_n(x_{n+1}+u)^{-1} \varphi(x_{n+1}-x_n)
$$
which is a direct consequence of the pentagon equation~\eqref{pentagon}.
\end{proof}

\begin{lemma}
\label{lem-int-f}
Let us set
$$
I_n(f;\bs\la,\bs\mu) = \int_{\R^n} f(x_n) \Psi^{(n)}_{\bs\la}(\bs x) \overline{\Psi^{(n)}_{\bs{\bar\mu}}(\bs x) } d\bs x,
$$
where $f$ is a meromorphic function of a single variable and $\bs\la, \bs\mu \in \C^n$. Then the following integral identity holds:
\begin{align*}
I_n(f;\bs\la,\bs\mu) = B_{n-1}(\bs\la,\bs\mu) \Big< f(x_n) e^{2\pi i (\la_n-\mu_n) x_n} &\varphi(x_n-x_{n-1}-\mu_n-c_b)^{-1} \Psi_{\bs{\la'}}^{(n-1)}(\bs{x'}), \\
&\varphi(x_n-x_{n-1}-\bar\la_n-c_b)^{-1} \Psi_{\bs{\bar\mu'}}^{(n-1)}(\bs{x'}) \Big>.
\end{align*}
where
$$
B_{n-1}(\bs\la,\bs\mu) = e^{2\pi ic_b\hr{\rho_{n-1}(\bs\la+\bs\mu)-\rho_n(\bs\la+\bs\mu)}} e^{\pi i \sum_{j=1}^{n-1} \hr{\la_j^2-\mu_j^2}} \prod_{j=1}^{n-1} \frac{\varphi(\mu_j-\la_n+c_b)}{\varphi(\la_j-\mu_n-c_b)}.
$$
\end{lemma}

\begin{proof}
Making use of Definition~\ref{def-whit}, we can express the product
$$
e^{2\pi ic_b\hr{\rho_n(\bs\la+\bs\mu)-\rho_{n-1}(\bs\la+\bs\mu)}} I_n(f;\bs\la,\bs\mu,\bs x)
$$
in the form
\begin{multline*}
\ha{f(x_n) Q_{n-1}^{(n)}(c_b-\la_n) \Psi_{\bs{\la'}}^{(n-1)}(\bs{x'}), Q_{n-1}^{(n)}(c_b-\bar\mu_n) \Psi_{\bs{\bar\mu'}}^{(n-1)}(\bs{x'})} \\
= \ha{f(x_n) Q_{n-1}^{(n)}(c_b-\bar\mu_n)^* Q_{n-1}^{(n)}(c_b-\la_n) \Psi_{\bs{\la'}}^{(n-1)}(\bs{x'}), \Psi_{\bs{\bar\mu'}}^{(n-1)}(\bs{x'})}.
\end{multline*}
Applying Lemma~\ref{conj-Q} and Corollary~\ref{Q-eigen-prop}, we get
\begin{multline*}
I_n(f;\bs\la,\bs\mu) = e^{2\pi ic_b\hr{\rho_{n-1}(\bs\la+\bs\mu)-\rho_n(\bs\la+\bs\mu)}} \prod_{j=1}^{n-1} \frac{\varphi(\mu_j-\la_n+c_b)}{\varphi(\la_j-\mu_n-c_b)} \\
\cdot \ha{ f(x_n) e^{2\pi i (\la_n-\mu_n)x_n} Y_{n-1}(x_n-c_b-\mu_n) \Psi_{\bs{\la'}}^{(n-1)}(\bs{x'}), Y_{n-1}(x_n-c_b-\bar\la_n)\Psi_{\bs{\bar\mu'}}^{(n-1)}(\bs{x'}) }.
\end{multline*}
The unitarity of the Dehn twist operator $D_{n-1}$ allows us to replace both operators $Y_{n-1}$ by the composites $D_{n-1}Y_{n-1}$. The result then follows from Lemma~\ref{y-dehn} and Proposition~\ref{dehn-eigen-prop}.
\end{proof}

In what follows, it will be convenient to express the prefactor $B_{n-1}(\bs\la,\bs\mu)$ in terms of the $c$-function. Recalling the relation between the $c(z)$ and $\varphi(z)$ from Section~\ref{subsec-c}, a straightforward calculation shows that
\beq
\label{B}
B_{n-1}(\bs\la, \bs\mu) = e^{\pi i\hr{\underline\la'\mu_n - \underline\mu'\la_n}} e^{\frac{\pi i (n-1)}{2}\hr{\la_n^2-\mu_n^2}} \prod_{j=1}^{n-1} \hr{c(\la_j-\mu_n) c(\la_n-\mu_j) e^{\frac{\pi i}{2}\hr{\la_j^2-\mu_j^2}}}.
\eeq

\subsection{The orthogonality integral}
\label{subsec-ort-int}

Here we evaluate the following integral, which constitutes the main step in the proof of orthogonality relation for the $b$\,-Whittaker functions:
$$
\nonumber O_n(\bs\la,\bs\mu,z) = \int_{\R^{n}}e^{2\pi i zx_n}\Psi_{\la}^{(n)}(\bs x) \overline{\Psi_{\bar\mu}^{(n)}(\bs x)} d\bs x.
$$

\begin{prop}
\label{orthog-int}
We have
$$
O_n(\bs\la,\bs\mu,z) = \delta(z+\bs{\underline\la}-\bs{\underline\mu}) e^{\frac{\pi i}{2}\hr{\bs{\underline\mu}^2-\bs{\underline\la}^2 + n\sum_{j=1}^n \hr{\la_j^2-\mu_j^2}} } c(\bs{\underline{\la}}-\bs{\underline{\mu}})^{-1} \prod_{j,k=1}^nc(\la_j-\mu_k).
$$
\end{prop}

\begin{proof}
We first observe that in the case $n=1$, the assertion of the Proposition reduces to the identity
$$
\int_{\R}e^{2\pi ix_1(z+\la_1-\mu_1)}dx_1 = \delta(z+\la_1-\mu_1).
$$
In the general case, we use Lemma~\ref{lem-int-f} to get
$$
O_n(\bs\la,\bs\mu,z) = B_{n-1}(\bs\la, \bs\mu) \breve O_n(\bs \la,\bs\mu,u),
$$
where
\begin{align*}
&\breve O_n(\bs\la,\bs\mu,z) = \ha{{e^{2\pi i (\la_n-\mu_n+z)x_n}} \frac{\varphi(x_n-x_{n-1}-\la_n+c_b)}{\varphi(x_n-x_{n-1}-\mu_n-c_b)} \Psi_{\la'}^{(n-1)}(\bs{x'}),  \Psi_{\bar\mu'}^{(n-1)}(\bs{x'})}.
\end{align*}
Applying the pentagon identity~\eqref{beta-1} to take the integral over $x_n$ and rewriting the result in terms of $c$-functions, we find that
$$
\breve O_n(\bs\la,\bs\mu,z) = e^{\pi i(\la_n+\mu_n)(\la_n-\mu_n+z)} c(\la_n-\mu_n) \frac{c(\mu_n-\la_n-z)}{c(-z)} O_{n-1}(\bs{\la'},\bs{\mu'},z+\la_n-\mu_n).
$$
Combining the above expression with the formula~\eqref{B} and continuing by induction, we obtain the desired statement.
\end{proof}

\subsection{The Cauchy--Littlewood identity}
\label{subsec-CL}

Let us consider an integral of the form
\beq
\label{C-int}
C_n(\bs\la,\bs\mu,u) = \int_{\R^n} \varphi(u-x_n)^{-1} \Psi^{(n)}_{\bs\la}(\bs x) \overline{\Psi^{(n)}_{\bs{\bar\mu}}(\bs x) } d\bs x.
\eeq

\begin{prop}
\label{cauchy-littlewood}
The following modular $b$\,-analog of the Cauchy-Littlewood identity holds:
$$
C_n(\bs\la,\bs\mu,u) = e^{\pi i(2u+c_b-\underline{\bs\mu})(\underline{\bs\la}-\underline{\bs\mu})} e^{\frac{\pi in}{2}\sum_{j=1}^n \hr{\la_j^2 - \mu_j^2}} \prod_{j,k=1}^{n} c(\la_j-\mu_k).
$$
\end{prop}

\begin{proof}
First, let us treat the $n=1$ case. Shifting the contour of integration, we obtain
$$
C_1(\bs\la,\bs\mu,u) = \int_\R \varphi(u-x)^{-1} e^{2\pi i(\la-\mu)x} dx = e^{2\pi i(u+c_b)(\la-\mu)} \int_{\R+i0} \frac{e^{2\pi i x(\mu-\la)}}{\varphi(x-c_b)} dx.
$$
Now, using the Fourier transform formula~\eqref{Fourier-1}, we get
\begin{align*}
C_1(\bs\la,\bs\mu,u)
&= \zeta \varphi(\mu-\la+c_b) e^{2\pi i(u+c_b)(\la-\mu)} \\
&= c(\la-\mu) e^{\pi i(2u+c_b)(\la-\mu)} e^{\frac{\pi i}{2} (\la-\mu)^2}.
\end{align*}

In order to treat the general case, let us introduce the following shorthands:
$$
v = x_n-x_{n-1}-\mu_n
\qquad\text{and}\qquad
w = x_n-x_{n-1}-\la_n.
$$
Then by Lemma~\ref{lem-int-f}, we have
$$
C_n(\bs\la,\bs\mu,u) = B_{n-1}(\bs\la, \bs\mu) \breve C_n(\bs\la,\bs\mu,u),
$$
where
$$
\breve C_n(\bs \la,\bs\mu,u) = \ha{ \frac{e^{2\pi i (\la_n-\mu_n)x_n}}{\varphi(u-x_n)} \varphi(v-c_b)^{-1} \Psi_{\bs{\la'}}^{(n-1)}(\bs{x'}), \varphi(\bar w-c_b)^{-1} \Psi_{\bs{\bar\mu'}}^{(n-1)}(\bs{x'}) }.
$$
Invoking Proposition~\ref{prop-phi-inv}, we can rewrite this quantity as
$$
\breve C_n(\bs\la,\bs\mu,u) = \ha{ \frac{e^{2\pi i (\la_n-\mu_n)x_n}}{\varphi(u-x_n)} \varphi(p_n+v) \Psi_{\la'}^{(n-1)}(\bs{x'}), \varphi(p_n+\bar w) \Psi_{\bar\mu'}^{(n-1)}(\bs{x'}) }.
$$
Now, using the fact that
$$
\varphi(p_n+w)^{-1} e^{2\pi i(v-w)x_n} = e^{2\pi i(v-w)x_n} \varphi(p_n+v)^{-1},
$$
along with the pentagon identity in the form
$$
\varphi(p_n+v)^{-1} \varphi(u-x_n)^{-1} \varphi(p_n+v) = \varphi(u-x_n)^{-1}\varphi(p_n-x_n+u+v)^{-1},
$$
and the fact that $\Psi_{\bs{\bar\mu}}^{(n-1)}(\bs{x'})$ is independent of $x_n$, we find that
\begin{align*}
\breve C_n(\bs\la,\bs\mu,u)=
\ha{\frac{e^{2\pi i (\la_n-\mu_n)x_n}}{\varphi(u-x_n)} \varphi(u-x_{n-1}-\mu_n)^{-1} \Psi_{\bs{\la'}}^{(n-1)}(\bs{x'}),  \Psi_{\bs{\bar\mu'}}^{(n-1)}(\bs{x'})}.
\end{align*}
At this point the integral over $x_n$ can be taken, to yield
\begin{align*}
\breve C_n(\bs\la,\bs\mu,u)
&= \zeta \varphi(\mu_n-\la_n+c_b) e^{2\pi i (\la_n-\mu_n)(u+c_b)} C_{n-1}(\bs\la',\bs\mu',u-\mu_n) \\
&= c(\la_n-\mu_n) e^{\frac{\pi i}{2}(\la_n-\mu_n)^2} e^{\pi i (\la_n-\mu_n)(2u+c_b)} C_{n-1}(\bs\la',\bs\mu',u-\mu_n).
\end{align*}
Using formula~\eqref{B} and continuing by induction, we arrive at the desired statement.
\end{proof}

\section{Unitarity of the Whittaker transform.}
\label{thm-proof}

In this section we present the proof of Theorem~\ref{main-thm}. The overall structure of the argument is modeled on the one used in \cite{Kas01}, \cite{FT15} to establish the result in the $\gl_2$ case. We begin by outlining the logic of the proof of the assertion (1). As the first step, we prove Theorem~\ref{completeness} which implies that $\Wc$ maps the dense domain $\Fc_n$ isometrically into $L^2_{\mathrm{sym}}(\R^n,m(\bs\la)d\bs\la)$, and thus admits an isometric extension to $L^2(\R^n)$, thereby proving the completeness relation $\Wc^*\Wc=\Id$. To finish the proof we must show that the image of $L^2(\R^n)$ under $\Wc$ coincides with $L^2_{\mathrm{sym}}(\R^n,m(\bs\la)d\bs\la)$, which we show by proving the orthogonality relation $\Wc\Wc^*=\mathrm{Id}$. We establish this relation in Theorem~\ref{orthogonality}, thereby completing the proof of part (1) of Theorem~\ref{main-thm}.

For part (2), Proposition~\ref{toda-eigen-prop} implies that $\Wc$ satisfies the claimed intertwining relation on the space of test functions of the form
$$
\prod_{k=1}^{n-1}\varphi(x_{k+1}-x_k)\cdot\left(\Fc_0\right)^{\otimes n}.
$$
The result now follows in the standard way, see for example the proof of Theorem 2.6 in~\cite{Gon05}, from the density of this subspace in the Fock--Goncharov Schwartz space, which in turn is a consequence of Theorem~2.3 in {\it op. cit.}

\begin{remark} In the non $q$-deformed case, the analog of Theorem~\ref{main-thm} is due to Semenov-Tian-Shansky \cite{Sem94}, and has been subsequently re-proved by Kozlowski \cite{Koz15} in the framework of the quantum inverse scattering method. 
\end{remark} 
\subsection{Completeness relation for $b$\,-Whittaker functions}
\label{sect-completeness}

The goal of this section is to prove the following theorem.

\begin{theorem}
\label{completeness}
For all test functions $f\in\Fc_n$, we have
$$
\int_{\R^n}\hr{\int_{\R^n} f(\bs x) \Psi_{\bs\la}^{(n)}(\bs x)d\bs x} \overline{\Psi_{\bs\la}^{(n)}(\bs y)} \mf(\bs\la) d\bs\la  = f(\bs y).
$$
\end{theorem}

The proof of Theorem~\ref{completeness} is based on an induction over the rank $n$. We shall break the proof of the induction step into several Lemmas. 

\begin{lemma}
\label{induction-completeness}
For all test functions $f\in\Fc_{n+1}$, we have
\begin{align*}
\int_{\R^{n+1}}\hr{\int_{\R^{n+1}} f(\bs x)\Psi_{\bs\la}^{(n+1)}(\bs x)d\bs x} \overline{\Psi_{\bs\la}^{(n+1)}(\bs y)} \mf_{n+1}(\la) d\bs\la = \int_{\R^{n+1}} \delta(\underline{\bs x}-\underline{\bs y}) f(\bs x)\Ic^{(n)}(\bs x, \bs y) d\bs{x},
\end{align*}
where
\beq
\label{In-int}
\begin{aligned}
\Ic^{(n)}(\bs x, \bs y) = &\int e^{2\pi i (2c_b+ib) \underline{\bs z}} e^{-2\pi i (2c_b \underline{\bs t} + ib \underline{\bs s})} \prod_{k=1}^{n-1} \frac{\varphi(z_{k+1}-t_k)}{\varphi(z_{k+1}-s_k)} \\
\cdot &\prod_{k=1}^n \frac{\varphi(t_k-z_k+c_b) \varphi(s_k-y_k+c_b)\varphi(y_{k+1}-s_k)}{\varphi(s_k-z_k-c_b) \varphi(t_k-x_k-c_b) \varphi(x_{k+1}-t_k)} d\bs z d\bs t d\bs s.
\end{aligned}
\eeq
\end{lemma}	

\begin{proof}
The Lemma is proved by induction on the rank $n$. More specifically, in view of the rapid decay of $\Wc[f]$ we can use the Fubini theorem to split the integral over $\bs \la$ into that over $\bs{\la'}$ followed by the one over $\la_{n+1}$. By Definition~\ref{def-whit} and Proposition~\ref{Q-eigen-prop} we have
$$
e^{2\pi ic_b \rho_n(\bs\la)} \prod_{k=1}^n \varphi(\la_k+u)^{-1} \Psi_{\bs\la}^{(n+1)}(\bs x) = Q_n^{n+1}(c_b-\la_{n+1}) Q_n(u)^{-1} \Psi_{\bs{\la'}}^{(n)}(\bs{x'}).
$$
Using integral kernels from Corollaries~\ref{cor-Q-inv} and~\ref{rec-act}, we see that the right hand side of the above equality takes form
$$
\int e^{2\pi i \la_{n+1} \hr{\underline{\bs x} - \underline{\bs z}}} e^{4\pi ic_b\hr{\underline{\bs z} -\underline{\bs t}}} \prod_{k=1}^n \frac{\varphi(t_k-z_k+c_b) }{\varphi(t_k-x_k-c_b)\varphi(x_{k+1}-t_k)} \prod_{k=1}^{n-1}\varphi(z_{k+1}-t_k) \Psi_{\bs{\la'}}^{(n)}(\bs z) d\bs z d\bs t,
$$
for $u = c_b-\la_{n+1}$, and
$$
\int e^{2\pi i \la_{n+1}\hr{\underline{\bs y} -\underline{\bs w}}} e^{2\pi b\hr{\underline{\bs s} -\underline{\bs w}}} \prod_{k=1}^n \frac{\varphi(s_k-w_k+c_b)}{\varphi(s_k-y_k-c_b)\varphi(y_{k+1}-s_k)} \prod_{k=1}^{n-1}\varphi(w_{k+1}-s_k) \Psi_{\bs{\la'}}^{(n)}(\bs w) d\bs w d\bs s,
$$
for $u = \Delta_b-\la_{n+1}$. Now, recalling Corollary~\ref{cor-measure}, using the induction hypothesis
$$
\int_{\R^n} \Psi_{\la'}^{(n)}(\bs z) \overline{\Psi_{\la'}^{(n)}(\bs w)} \mf_{n}(\bs\la') d\bs\la' = \delta(\bs z- \bs w),
$$
and the Fourier inversion formula
we obtain the assertion of the Lemma. 
\end{proof}

Thus the problem is reduced to understanding the distribution defined by the kernel $\Ic^{(n)}(\bs x, \bs y)$. Indeed, to finish the proof of Theorem~\ref{completeness}, it will suffice to prove

\begin{prop}
\label{completeness-prop}
For all test functions $f\in\Fc_{n+1}$, we have
$$
\int_{\R^{n+1}} \delta(\underline{\bs x}-\underline{\bs y}) f(\bs x)\Ic^{(n)}(\bs x, \bs y) d\bs{x} = f(\bs y).
$$
\end{prop}

\begin{proof}
Let is introduce a parameter $\ep>0$, and consider the regularized kernel
$$
\Ic^{(n)}_\ep(\bs x, \bs y) = \Ic^{(n)}(\bs x[-\ep], \bs y[\ep])
\qquad\text{where}\qquad
\bs x[u] = (x_1+iu, \dots, x_{n+1}+iu).
$$
By the dominated convergence theorem and the rapid decay of а $f(\bs x)$, it suffices to show that
$$
\lim_{\ep\to0} \mathcal C_\ep[f](\bs y) = f(\bs y)
\qquad\text{for}\qquad
\mathcal C_\ep[f](\bs y) = \int_{\R^{n+1}} \delta(\underline{\bs x}-\underline{\bs y}) f(\bs x)\Ic_\ep^{(n)}(\bs x, \bs y) d\bs{x}
$$
in order to prove the Proposition~\ref{completeness-prop}. In what follows it will prove useful to consider integrals
\begin{align*}
\hat I_k &= \int_{\R} \frac{\varphi(s-y_k+c_b-i\ep)\varphi(y_{k+1}-s+i\ep)}{\varphi(s-x_k-c_b+i\ep) \varphi(x_{k+1}-s-i\ep)} ds, \\
\hat J_k &= \int_{\R} \frac{\varphi(s-y_k+\frac{ib^{-1}}{2}-i\ep)\varphi(y_{k+1}-s+\frac{ib} {2}+i\ep)}{\varphi(s-x_k-\frac{ib^{-1}}{2}+i\ep) \varphi(x_{k+1}-s-\frac{ib} {2}-i\ep)} ds, \\
I_k &= \int_{\R}\frac{\varphi(s-y_k+c_b-i\ep)\varphi(y_{k+1}-s+i\ep)}{\varphi(s-x_k-c_b+i\ep) \varphi(x_{k+1}-s-i\ep)}e^{2\pi b (s-x_k+i\ep)} ds, \\
J_k &= \int_{\R}\frac{\varphi(s-y_k+c_b-i\ep)\varphi(y_{k+1}-s+i\ep)}{\varphi(s-x_k+\Delta_b+i\ep) \varphi(x_{k+1}-s-ib-i\ep)}e^{2\pi b (x_{k+1}-s-\frac{ib}{2}-i\ep)} ds.
\end{align*}
For sufficiently small $\ep>0$, each of these integrals is absolutely convergent and the integration contour $\R$ separates the upwards and downwards sequences of poles of each integrand. For concreteness, let us suppose that $b\in (0,1)$, so that $b^{-1}>b$. (Were this not the case, we could have used the alternative expression for the measure obtained by swapping $b$ and $b^{-1}$.)  Then the integrand of $J_k$ is analytic in the strip of width $\frac{ib}{2}$ around $\mathbb{R}$, so we can shift its integration contour by $-\frac{ib}{2}$ 
to write
$$
\hat J_k = \int_{\R} \frac{\varphi(s-y_k+c_b-i\ep)\varphi(y_{k+1}-s+i\ep)}{\varphi(s-x_k+\Delta_b+i\ep) \varphi(x_{k+1}-s-ib-i\ep)} ds,
$$
and
$$
 J_k = e^{2\pi b (x_{k+1}-i\ep)}\int_{\R} \frac{\varphi(s-y_k+c_b-i\ep)\varphi(y_{k+1}-s+i\ep)}{\varphi(s-x_k+\Delta_b+i\ep) \varphi(x_{k+1}-s-ib-i\ep)} e^{-2\pi bs} ds.
$$

\begin{lemma}
\label{diff-IJ}
We have
\beq
\label{I-J-hat}
\hat I_k - \hat J_k =  I_k +  J_k.
\eeq
\end{lemma}

\begin{proof}
The functional equation~\eqref{eqn-func} imply that
$$
\hat J_k = \int_{\R} \frac{\varphi(s-y_k+c_b-i\ep)\varphi(y_{k+1}-s+i\ep)}{\varphi(s-x_k-c_b+i\ep) \varphi(x_{k+1}-s-i\ep)} \hr{ \frac{1+e^{2\pi b (s-x_k-\frac{ib^{-1}}{2}+i\ep)}}{1+e^{2\pi b (x_{k+1}-s-\frac{ib}{2}-i\ep)}}}ds.
$$
Using the equality
$$
\frac{1+A}{1+B} = 1 + A - B \frac{1+A}{1+B}
$$
with
$$
A = e^{2\pi b (s-x_k-\frac{ib^{-1}}{2}+i\ep)}
\qquad\text{and}\qquad
B = e^{2\pi b (x_{k+1}-s-\frac{ib}{2}-i\ep)}
$$
and noting that
$$
e^{2\pi b (s-x_k-\frac{ib^{-1}}{2}+i\ep)} = - e^{2\pi b (s-x_k+i\ep)}
$$
we arrive at
$$
\hat J_k = \hat I_k -  I_k - J_k.
$$
\end{proof}

As a consequence, we have

\begin{lemma}
The kernel $\Ic^{(n)}_\ep(\bs x,\bs y) $ can be written as
\beq
\label{In-sum}
\Ic^{(n)}_\ep(\bs x,\bs y) = \sum_{k=0}^n I_1 \dots I_k J_{k+1} \dots J_n.
\eeq
\end{lemma}
\begin{proof}
We can use the distributional identities in Lemma~\ref{dist-identities} to take consecutive integrals over $z_1, z_2, \dots, z_n$ in the formula~\eqref{In-int}. The result is expressed in terms of the integrals $I_k$, $\hat I_k$, $J_k$, and $\hat J_k$, and with the help of the equality~\eqref{I-J-hat} can be brought to the desired form.
\end{proof}

The integrands of $I_k$ and $J_k$ have poles that approach the integration contour $\R$ as $\ep\rightarrow 0$. However, as explained in~\cite{DF14} in the rank 1 case, we can use the pentagon identity to rewrite these integrals in a form that will enable us to pass to the limit $\ep\rightarrow 0$. 

\begin{lemma}
\
For $0<\delta<2\ep$, consider the integral
\beq
\label{Lk-def}
L_k = \int _{\R} e^{2\pi i (x_{k}-y_{k+1}+c_b-2i\ep)(z-i\delta)}\frac{\varphi(z-\Delta_b-i\delta)\varphi(-z+x_{k+1}-y_{k+1}+c_b-2i\ep+i\delta)}{\varphi(-z-c_b+i\delta)\varphi(z+x_k-y_k-\Delta_b-2i\ep-i\delta)}dz.
\eeq
Then we have
$$
\Ic^{(n)}_\ep = M_n \cdot L_1 \dots L_n, 
$$
where
$$
M_n = \left(1-e^{-4\pi b (n+1)i\ep}\right)\frac{\varphi(x_1-y_1+c_b-2i\ep)}{\varphi(x_{n+1}-y_{n+1}-\Delta_b-2i\ep)}.
$$
\end{lemma}

\begin{proof}
For any $\delta$ satisfying $0<\delta<2\ep$, we may use the pentagon identity to write 	 
$$
\frac{\varphi(s-y_k+c_b-i\ep)}{\varphi(s-x_k-c_b+i\ep)}=\zeta \varphi(x_k-y_k+c_b-2i\ep)\int_{\R} \frac{e^{-2\pi i (s-x_k+i\ep)(z-ib-i\delta)} \varphi(z-\Delta_b-i\delta)}{\varphi(z+x_k-y_k-\Delta_b-2i\ep-i\delta)}dz,
$$
where the integral converges absolutely and the contour separates the pole sequences of the integrand. Inserting this into the definitions of $I_k$ and $J_k$, and applying the pentagon identity again to integrate over $s$, we arrive at
$$
I_k = \frac{\varphi(x_k-y_k+c_b-2i\ep)}{\varphi(x_{k+1}-y_{k+1}+c_b-2i\ep)}L_k
$$
and
$$
J_k = e^{2\pi b(x_{k+1}-y_{k+1}-2i\ep)}\frac{\varphi(x_k-y_k-\Delta_b-2i\ep)}{\varphi(x_{k+1}-y_{k+1}-\Delta_b-2i\ep)}L_k
$$
where both integrals converge absolutely and the contours are pole-separating.

Plugging the above expressions into the equality~\eqref{In-sum} we obtain
$$
\Ic^{(n)}_\ep = M_n \cdot L_1 \dots L_n,
$$
where
$$
M_n = \sum_{k=0}^n \frac{\varphi(x_1-y_1+c_b-2i\ep)}{\varphi(x_{k+1}-y_{k+1}+c_b-2i\ep)} \frac{\varphi(x_{k+1}-y_{k+1}-\Delta_b-2i\ep)}{\varphi(x_{n+1}-y_{n+1}-\Delta_b-2i\ep)} e^{2\pi b \sum_{r=k+2}^{n+1} (x_r-y_r-2i\ep)}.
$$
By the functional equation~\eqref{eqn-func}, for $0\leq k\leq n$ we have
$$
\frac{\varphi(x_{k+1}-y_{k+1}-\Delta_b-2i\ep)}{\varphi(x_{k+1}-y_{k+1}+c_b-2i\ep)} = 1-e^{2\pi b(x_{k+1}-y_{k+1}-2i\ep)},
$$
and therefore the prefactor $M_n$ takes form
$$
M_n = \left(1-e^{2\pi b (\underline{\bs x} - \underline{\bs y}-2(n+1)i\ep)}\right)\frac{\varphi(x_1-y_1+c_b-2i\ep)}{\varphi(x_{n+1}-y_{n+1}-\Delta_b-2i\ep)}.
$$
In particular, under the balancing condition $\underline{\bs x} = \underline{\bs y}$, we obtain
$$
M_n = \hr{1-e^{-4\pi b (n+1)i\ep}} \frac{\varphi(x_1-y_1+c_b-2i\ep)}{\varphi(x_{n+1}-y_{n+1}-\Delta_b-2i\ep)},
$$
and the Lemma is proved.
\end{proof}

We now need to understand the behavior of the integrals $L_k$ as $\ep\to0$. Note that the integrand of $L_k$ has a pair of poles at $z=i\delta$ and $z=x_{k+1}-y_{k+1}-2i\ep+i\delta$, which approach the real line respectively from above and below as $\ep\to0$. 
Let us deform the contour of integration over $z$ to a contour $C_+$ obtained by deforming the contour $\R$ upward to cross the pole $z=i\delta$ and no others. In this way, we arrive at
\beq
\label{Lk-sum}
L_k = L_k^+ + R_k,
\eeq
where $L_k^+$ coincides with~\eqref{Lk-def} except that the integral is now taken over the contour $C_+$, and the contribution $R_k$ is the residue of $L_k$ at $z=i\delta$ multiplied by a factor of $2\pi i$. Recalling that
$$
\zeta \varphi(-\Delta_b) = ib,
$$
we obtain
$$
R_k = ib \frac{\varphi(x_{k+1}-y_{k+1}+c_b-2i\ep)}{\varphi(x_k-y_k-\Delta_b-2i\ep)}.
$$
Note that the integrand in $L_k^+$ now does not have poles approaching the integration contour $C^+$ as $\ep\to0$.
Rewriting each of the factors $L_k$ in the form~\eqref{Lk-sum}, we obtain
\begin{align}
\nonumber
\Cc_\ep[f](\bs y)
&= \sum_{S\subset[1,n]} \int_{\R^{n+1}} \delta(\underline{\bs x}-\underline{\bs y}) f(\bs x) M_n \prod_{j\in S} L^+_i \prod_{k\not\in S} R_k d \bs x \\
\label{C-ep}
&= \sum_{S\subset[1,n]} \int_{\R^{n}} f(\underline{\bs y}-{\underline{\bs{'\!x}}}, \bs{'\!x}) M_n \prod_{j\in S} L^+_i \prod_{k\not\in S} R_k d\bs{'\!x},
\end{align}
where we employ Notation~\ref{rebus}. We now observe that unless $S$ is empty, the corresponding summand in~\eqref{C-ep} vanishes in the limit $\ep\rightarrow 0$. Indeed, if $j\in S$, then we can deform the contours of integration for all $x_{k+1}$ with $k \in [1,n] \setminus S$ downward by a finite distance in the vicinity of $x_{k+1}=y_{k+1}$ to avoid the poles at $x_{k+1}=y_{k+1}+2i\ep$ coming from the factors $R_k$, while shifting the contour for $x_{j+1}$ upward in order to avoid the pole at $\underline{\bs y}-{\underline{\bs{'\!x}}} = y_1+2i\ep$ potentially coming from the factor $M_n$. Since the integrals~\eqref{C-ep} obtained after such a deformation converge absolutely for all small $\epsilon>0$ while having a prefactor $\hr{1-e^{-4\pi b (n+1)i\ep}}$, the corresponding contribution thus vanishes as $\ep\to0$.

We therefore see that the only nonzero contribution in~\eqref{C-ep} comes from the summand in which $S$ is the empty set. The latter contribution is given by
$$
 \hr{1-e^{-4\pi b (n+1)i\ep}} (ib)^n \int_{\R^{n}} f(\underline{\bs y}-{\underline{\bs{'\!x}}}, \bs{'\!x}) \prod_{j=1}^{n+1} \frac{\varphi(x_j-y_j+c_b-2i\ep)}{\varphi(x_j-y_j-\Delta_b-2i\ep)} d\bs{'\!x}.
$$
If we shift the contour of integration over $x_{n+1}$ upward by a small finite amount in the vicinity of the point $x_{n+1} = y_{n+1}$, thus crossing the pole $x_{n+1} = y_{n+1}+2i\ep$ and no others, we again obtain an integral that converges absolutely for all small $\ep>0$ and thus vanishes in the limit, plus a contribution
$$
\hr{1-e^{-4\pi b (n+1)i\ep}} (ib)^{n-1} \int_{\R^{n-1}} f(\underline{\bs{y'}}-{\underline{\bs{'\!x'}}}-2i\ep, \bs{'\!x'}, y_{n+1}+2i\ep) \prod_{j=1}^n \frac{\varphi(x_j-y_j+c_b-2i\ep)}{\varphi(x_j-y_j-\Delta_b-2i\ep)} d\bs{'\!x'}
$$
that comes from the corresponding residue. Deforming the rest of the contours in this fashion, we see that the only non-zero contribution to the limit of $\Cc_\ep[f]$ as $\ep\to$ is given by
$$
 \hr{1-e^{-4\pi b (n+1)i\ep}} \frac{\varphi(c_b-2(n+1)i\ep)}{\varphi(-\Delta_b-2(n+1)i\ep)} f(y_1-2ni\ep, y_2+2i\ep, \dots, y_{n+1}+2i\ep),
$$
and hence
$$
\Cc_\ep[f](\bs y) \to f(\bs y) \qquad\text{as}\qquad \ep \to 0.
$$
This completes the proof of Proposition~\ref{completeness-prop}, and therefore that of Theorem~\ref{completeness}. 
\end{proof}

\subsection{The orthogonality relation for the $b$\,-Whittaker functions. }
\label{sect-orthog}
In this section we shall complete the proof of Theorem~\ref{main-thm} by establishing 

\begin{theorem}
\label{orthogonality}
For all test functions $f\in\Fc_{n}$, we have
\beq
\label{orth-statement}
\int_{\R^{n}}\overline{\Psi_{\mu}^{(n)}(\bs x)}\hr{\int_{\R^{n}} \Psi_{\la}^{(n)}(\bs x)f(\bs \la)m(\bs\la) d\bs\la} d\bs x= \frac{1}{n!}\sum_{w\in S_{n}}f\hr{w(\bs\mu)}. 
\eeq
\end{theorem}

\begin{proof}
For small $\ep>0$, let us write 
$$
\bs\mu_\ep = (\mu_1+i\ep, \mu_2+2i\ep,\cdots, \mu_{n}+ni\ep),
$$
and consider the integral
$$
\mathcal{O}_\ep[f] = \int_{\R^{n}}\int_{\R^{n}} e^{n(n+1) \ep\pi x_n }\overline{\Psi_{\mu_\ep}^{(n)}(\bs x)} \Psi_{\la}^{(n)}(\bs x)f(\bs \la) m(\bs\la) d\bs\la d\bs x  . 
$$
By the dominated convergence theorem and the rapid decay of $\Wc^*[f]$ it follows that the left hand side of~\eqref{orth-statement} is recovered as $\lim_{\ep\to0}\mathcal{O}_\ep[f]$. Moreover, in view of the rapid decay in $\bs x$ of the regularized Whittaker function $e^{n(n+1) \ep\pi x_n }\Psi_{\bs\mu_\ep}^{(n)}(\bs x)$, the Fubini theorem can be applied to switch the order of integration in~\eqref{orth-statement}, and the integral over $\bs x$ can be taken with the help of Theorem~\ref{orthog-int}. Recalling the expression~\eqref{sklyanin-mes} for the measure $m(\bs\la)$ in terms of the $c$-function, we rewrite the result as
\beq
\label{ortho-dist}
\mathcal{O}_\ep[f] = \frac{1}{n!} c\hr{\frac{n(n+1)}{2}i\ep}^{-1} \int \delta(\bs{\underline{\la}}-\bs{\underline{\mu}}) f(\bs \la)\mathcal{E}_\ep(\bs\la,\bs\mu)\frac{\prod_{r,s=1}^{n}c(\la_r -\mu_s+si\ep)}{\prod_{j\neq k} c(\la_j-\la_k)  }  d\bs\la, 
\eeq
where
$$
\mathcal{E}_\ep(\bs\la,\bs\mu) = e^{\frac{\pi i}{2}\hr{\bs{\underline{\mu}}^2 - \bs{\underline{\la_\ep}}^2 +n\sum_{k=1}^n \hr{(\la_k+ki\ep)^2 - \mu_k^2}} }.
$$

Recall that the function $c(x)$ has a simple pole at $x=0$ with residue $-1$.  
Thus, arguing in a similar fashion to the proof of Proposition~\ref{completeness-prop} and deforming the contour of integration downwards to cross the simple poles of the integrand in~\eqref{ortho-dist}, we see that the only contribution as $\ep\to0$ are given by the residues of the poles approaching real line. This contribution constitutes the right hand side of~\eqref{orth-statement}. 
\end{proof}

\section{Hypergeometric integral evaluations and $b$\,-Whittaker functions}

In this section we explain how certain hyperbolic hypergeometric integral evaluations of Rains~\cite{Rai09,Rai10} can be derived naturally from the properties of the $b$\,-Whittaker functions. We also show how one can derive a hyperbolic generalization of Gustafson's integrals~\cite{Gus94}. In a similar fashion, the original Gustafson integrals have been recently considered from the perspective of the {\it XXX} spin magnet in~\cite{DM17,DMV17,DMV18}.

Consider an integral
\beq
\label{hypergeometric-int}
R(\bs\alpha,\bs\beta,u,v) =\int \psi^{(n+1)}_{\bs\alpha}(\bs x,u) \overline{\psi^{(n)}_{\bs\la}(\bs x)} {\psi^{(n)}_{\bs\la}(\bs y)}\overline{\psi^{(n+1)}_{\bs{\overline\beta}}(\bs y,v)}m(\bs\la)d\bs\la d\bs x d\bs y,
\eeq
where all contours of integration are taken to be $\R$. 
On the one hand, the integrals over $\bs x$ and $\bs y$ can be taken with the help of formula~\eqref{MB-rec}, to obtain
\beq
\label{pre-rains-LHS}
R(\bs\alpha,\bs\beta,u,v) =
e^{2\pi i \hr{u\underline{\bs\alpha}- v\underline{\bs\beta}}}\int e^{\pi i\underline{\bs\la}\hr{2v-2u +\underline{\bs\beta}-\underline{\bs\alpha}}}
\prod_{j=1}^{n+1} \prod_{k=1}^n c(\alpha_j-\la_k)c(\la_k-\beta_j) m({\bs\la})d{\bs\la}.
\eeq
Alternatively, one can integrate first over $\bs\la$ and apply the completeness relation for the $\mathfrak{gl}_n$ $b$\,-Whittaker functions   to write
\beq
\label{pre-rains-RHS}
R(\bs\alpha,\bs\beta,u,v) =
\int
\psi^{(n+1)}_{\bs\alpha}(\bs x,u) \overline{\psi^{(n+1)}_{\bs{\overline\beta}}(\bs x,v)}d\bs x. 
\eeq

\subsection{Gustafson's integral.}
Let us start by considering the specialization $u=v=0$, so that we have
$$
R(\bs\alpha,\bs\beta,0,0) =
\int e^{\pi i\underline{\bs\la}\hr{ \underline{\bs\beta}-\underline{\bs\alpha}}}
\prod_{j=1}^{n+1} \prod_{k=1}^n c(\alpha_j-\la_k)c(\la_k-\beta_j) m({\bs\la})d{\bs\la}.
$$
In this case, the argument used to prove Proposition~\ref{orthog-int} can be used to express the integral~\eqref{pre-rains-RHS} as

$$
R(\bs\alpha,\bs\beta,0,0) = B_n(\bs\alpha, \bs\beta) e^{\frac{n\pi i }{2}\sum_{j=1}^{n+1}\hr{\beta_j^2-\alpha_j^2}} \ha{ \frac{\varphi(-x_n-\alpha_{n+1}+c_b)}{\varphi(-x_n-\beta_{n+1}-c_b)} \Psi_{\bs{\alpha'}}^{(n)}(\bs x),  \Psi_{\bs{\overline\beta'}}^{(n)}(\bs x)}.
$$
This integral can be calculated by first using the pentagon identity to write
$$
\frac{\varphi(-x_n-\alpha_{n+1}+c_b)}{\varphi(-x_n-\beta_{n+1}-c_b)} = \zeta  \varphi(\beta_{n+1}-\alpha_{n+1}-c_b) \int \frac{\varphi(t+c_b) e^{2\pi i t(x_n+\beta_{n+1})} dt}{\varphi(t+\beta_{n+1}-\alpha_{n+1}-c_b)},
$$
and then applying Proposition~\ref{orthog-int}. Putting everything together, we obtain the following $b$\,-analog of an integral evaluation discovered by Gustafson, see~\cite[Theorem 5.1]{Gus94}.

\begin{prop}
We have
$$
\int e^{\pi i\underline{\bs\lambda}\hr{ \underline{\bs\beta}-\underline{\bs\alpha}}}
\prod_{j=1}^{n+1} \prod_{k=1}^n c(\alpha_j-\lambda_k)c(\lambda_k-\beta_j) m({\bs\lambda})d{\bs\lambda} = e^{\pi i \sum_{r<s}(\beta_r\beta_s - \alpha_r\alpha_s)} \frac{\prod_{j,k=1}^{n+1}c(\alpha_j-\beta_k)}{c(\bs{\underline{\beta}}-\bs{\underline{\alpha}})}.
$$
\end{prop}	

\subsection{Rains' integral. }
In a similar way, we can give an alternative derivation of the following hyperbolic hypergeometric integral evaluation first proved by Rains, see~\cite[Corollary 4.2]{Rai10} along with~\cite[Theorem 4.6]{Rai09}.

\begin{prop}
The integral identity
\beq
\label{rains-integral}
\int_{\hc{\underline{\bs \la}=0}} 
\prod_{j=1}^{n+1} \prod_{k=1}^n c(\alpha_j-\la_k)c(\la_k-\beta_j) m({\bs\la})d{\bs\la} =  \prod_{j,k=1}^{n+1}c(\alpha_j-\beta_k)\prod_{r=1}^{n+1}c(\underline{\bs \alpha}-\alpha_r)c(\beta_r-\underline{\bs \beta})
\eeq
holds under the balancing condition 
\beq
\label{rains-balancing}
\underline{\bs \alpha} - \underline{\bs \beta} = 2c_b.
\eeq
\end{prop}

\begin{proof}
The strategy is to calculate
$$
\int e^{2\pi i \underline{\bs \beta} v}R(\bs\alpha,\bs\beta,0,v)dv
$$	
in two different ways. On the one hand, the formula~\eqref{pre-rains-LHS} leads us directly to the integral on the left hand side of~\eqref{rains-integral}. Alternatively, we can use the expression~\eqref{pre-rains-RHS} for $R(\bs\alpha,\bs\beta,0,v)$. Note that
$$
\int e^{2\pi i \underline{\bs \beta} v}R(\bs\alpha,\bs\beta,0,v)dv = e^{\frac{n\pi i }{2}\sum_{j=1}^{n+1}\hr{\alpha_j^2-\beta_j^2}} \int e^{2\pi i \underline{\bs\beta}v} I_{n+1}(\bs\alpha, \bs\beta; -v, v) dv,
$$
where
\beq
\label{I-integral}
I_{n+1}(\bs\alpha, \bs\beta; u,v) = \int \Psi^{(n+1)}_{\bs\alpha}(\bs{x'}, u+v) \overline{\Psi^{(n+1)}_{\bs{\overline\beta}}(\bs x, v)} d\bs x.
\eeq
Arguing as in the proof of the Lemma~\ref{lem-int-f}, we see that the rescaled integral
$$
I_{n+1}(\bs\alpha, \bs\beta; u,v) \cdot B_n(\bs\alpha, \bs\beta)^{-1} e^{2\pi i(\beta_{n+1}v - \alpha_{n+1}v - \alpha_{n+1}u)}
$$
is given by the scalar product
$$
\ha{\frac{\varphi(v-x_n-p_n)}{\varphi(u+v-x_n-p_n)} \frac{\Psi^{(n)}_{\bs{\alpha'}}(\bs x)}{\varphi(u+v-x_n-\beta_{n+1}+c_b)}, \frac{\Psi^{(n)}_{\bs{\overline\beta'}}(\bs x)}{\varphi(v-x_n-\overline\alpha_{n+1}+c_b)}}.
$$
Using the pentagon identity in the integral form
$$
\frac{\varphi(v-x_n-p_n)}{\varphi(u+v-x_n-p_n)} = \zeta \varphi(-u-c_b)\int e^{2\pi i t(p_n+x_n-u-v-c_b)}\frac{\varphi(t+c_b)}{\varphi(t-u-c_b)} dt,
$$
and pulling $e^{2\pi i tp_n}$ all the way to the right, we can rewrite the above scalar product as
$$
\zeta \varphi(-u-c_b) \int \frac{\varphi(t+c_b) \varphi(v-x_n-\alpha_{n+1}+c_b)e^{\pi i t^2} e^{2\pi it(x_n-u-v-c_b)}}{\varphi(t-u-c_b) \varphi(u+v-x_n-t-\beta_{n+1}-c_b)} I_n(\bs{\alpha'}, \bs{\beta'}; t, x_n) dt dx_n,
$$
which gives us a recurrence relation for the integral~\eqref{I-integral}. Continuing by induction, we get\begin{align*}
I_{n+1}(\bs\alpha, \bs\beta; u, v) = \zeta^n S_n(\bs\alpha, \bs\beta) \int \frac{\varphi(t_1+c_b)}{\varphi(t_{n+1}+c_b)} \prod_{j=1}^{n+1} e^{2\pi ix_j(\alpha_j - \beta_j + t_j - t_{j-1})} e^{2\pi i t_j \alpha_j}& \\
\prod_{j=1}^n \frac{\varphi(t_{j+1}-t_j+c_b) \varphi(x_{j+1}-x_j-\alpha_{j+1}+c_b)}{\varphi(x_{j+1}-x_j+t_{j+1}-t_j-\beta_{j+1}-c_b)} dt_j dx_j&,
\end{align*}
where we set
\beq
\label{cond-tx}
t_0 = 0, \qquad t_{n+1} = u, \qquad x_{n+1} = v,
\eeq
and
$$
S_n(\bs\alpha, \bs\beta) = \prod_{j=1}^n B_j(\alpha_1, \dots, \alpha_{j+1}, \beta_1, \dots, \beta_{j+1}).
$$

The next step of the calculation is to use the pentagon identity again to successively take the integrals over $x_1,x_2, \dots x_n$. After doing so we obtain an expression
\begin{align*}
S_n(\bs\alpha, \bs\beta) \prod_{j<k} e^{2\pi i\alpha_k(\beta_j-\alpha_j)} e^{2\pi ix_{n+1}(\bs{\underline\alpha} - \bs{\underline\beta})} e^{2\pi i\alpha_{n+1}t_{n+1}} \frac{\varphi(t_{n+1}+\bs{\underline\alpha}-\bs{\underline\beta}-c_b)}{\varphi(t_{n+1}+c_b)}& \\
\cdot \int \prod_{j=1}^{n+1} \frac{\varphi(t_{j+1}-t_j+c_b)}{\varphi(t_{j+1}-t_j+\alpha_{j+1}-\beta_{j+1}-c_b)} \prod_{j=1}^n e^{2\pi i t_j(\alpha_j-\alpha_{j+1})} dt_j&,
\end{align*}
where we once again use the conventions~\eqref{cond-tx}. The balancing condition~\eqref{rains-balancing} ensures the cancellation of the first two dilogarithms in the right hand side of the latter equality. Now, we finally impose the condition $u=-v$, and use the pentagon identity to successively integrate over $v, t_n, t_{n-1}, \dots, t_1$. In this way we arrive at the formula
$$
\zeta^{-(n+1)} S_n(\bs\alpha, \bs\beta) \prod_{j<k} e^{2\pi i\alpha_k(\beta_j-\alpha_j)} \prod_{j=1}^{n+1} \frac{\varphi(\bs{\underline\alpha}-\beta_j-c_b)}{\varphi(\bs{\underline\alpha}-\alpha_j-c_b) \varphi(\alpha_j-\beta_j-c_b)}.
$$
Rewriting the latter expresssion in terms of $c$-functions and recalling the balancing condition~\eqref{rains-balancing}, we obtain the desired formula~\eqref{rains-integral}.
\end{proof}

\begin{remark}
Let
$$
R_{n}^m = \int_{\hc{\underline{\bs \la}=0}} \prod_{j=1}^{n+m+1} \prod_{k=1}^n c(\alpha_j-\la_k)c(\la_k-\beta_j) m({\bs\la})d{\bs\la}.
$$
In addition to the evaluation of $R_n^0$, Rains also derived certain $(A_n,A_m)$ transformation identities relating $R_n^m$ and $R_m^n$, under an appropriate balancing condition. Since the integral $R_n^m$ can be obtained from~\eqref{hypergeometric-int} by applying $\prod_{j=n+2}^{n+m+1} Q_n(c_b-\beta_j)$ to $\psi^{(n)}_{\bs\la}(\bs y)$ and $\prod_{j=n+2}^{n+m+1} Q_n(-c_b-\overline\alpha_j)$ to $\psi^{(n)}_{\bs\la}(\bs x)$, it is natural to expect that the $(A_n,A_m)$ transformations can be given a representation-theoretic proof using the $b$\,-Whittaker functions. We save the details for a future work. 
\end{remark}

\bibliographystyle{alpha}

\begin{thebibliography}{HKKR00}

\bibitem[BKP18a]{BKP18a}
{O.\,Babelon, K.\,Kozlowski, V.\,Pasquier.}
``Baxter operator and Baxter equation for the $q$-Toda and Toda$_2$ chains.''
{\emph arXiv:1803.06196} (2018).

\bibitem[BKP18b]{BKP18b}
{O.\,Babelon, K.\,Kozlowski, V.\,Pasquier.}
``Solution of Baxter equation for the $q$-Toda and Toda$_2$ chains by NLIE.''
{\emph arXiv:1804.01749} (2018).

\bibitem[COSZ14]{COSZ14}
{I.\,Corwin, N.\,O’Connell, T.\,Sepp\"al\"ainen, N.\,Zygouras.}
``Tropical combinatorics and Whittaker functions.''
\emph{Duke Mathematical Journal} 163, no. 3 (2014): 513-563.

\bibitem[DF14]{DF14}
{S.\,Derkachev, L.\,Faddeev.}
``3j-symbol for the modular double $SL_q(2,\R)$ revisited.''
\emph{Journal of Physics: Conference Series} 532, no. 1 (2014).

\bibitem[DKM01]{DKM01}
{S.\,Derkachev, G.\,Korchemsky, A.\,Manashov.}
``Noncompact Heisenberg spin magnets from high-energy QCD: I. Baxter Q-operator and Separation of Variables.''
\emph{Nuclear Physics B} 617, no. 1-3 (2001): 375-440.

\bibitem[DKM03]{DKM03}
{S.\,Derkachev, G.\,Korchemsky, A.\,Manashov.}
``Separation of variables for the quantum $SL(2,\R)$ spin chain.''
\emph{Journal of High Energy Physics} 2003, no. 07 (2003): 047.

\bibitem[DKM18]{DKM18}
{S.\,Derkachev, K.\,Kozlowski, A.\,Manashov.}
``On the separation of variables for the modular {\it XXZ} magnet and the lattice Sinh-Gordon models.''
\emph{arXiv:1806.04487} (2018).

\bibitem[DM14]{DM14}
{S.\,Derkachev, A.\,Manashov.}
``Iterative construction of eigenfunctions of the monodromy matrix for $SL(2,\C)$ magnet.''
\emph{Journal of Physics A: Mathematical and Theoretical} 47, no. 30 (2014): 305204.

\bibitem[DM17]{DM17}
{S.\,Derkachev, A.\,Manashov.}
``Spin chains and Gustafson’s integrals.''
\emph{Journal of Physics A: Mathematical and Theoretical} 50, no. 29 (2017): 294006.

\bibitem[DMV17]{DMV17}
{S.\,Derkachev, A.\,Manashov, P.\,Valinevich.}
``Gustafson integrals for spin magnet.''
Journal of Physics A: Mathematical and Theoretical 50, no. 29 (2017): 294007.

\bibitem[DMV18]{DMV18}
{S.\,Derkachev, A.\,Manashov, P.\,Valinevich.}
``$SL(2,\C)$ Gustafson Integrals."
\emph{Symmetry, Integrability and Geometry: Methods and Applications} 14 (2018): 030.

\bibitem[Eti99]{Eti99}
{P.\,Etingof.}
``Whittaker functions on quantum groups and $q$-deformed Toda operators.''
\emph{American Mathematical Society Translations: Series 2} 194 (1999): 9-26.

\bibitem[Fad99]{Fad99}
{L.\,Faddeev.}
``Modular Double of Quantum Group.''
\emph{arXiv:math/9912078} (1999).

\bibitem[FG06]{FG06}
{V.\,Fock, A.\,Goncharov.}
``Moduli spaces of local systems and higher Teichm\"uller theory.''
\emph{Publications Math\'ematiques de l'IH\'ES} 103 (2006): 1-211.

\bibitem[FG09]{FG09}
{V.\,Fock, A.\,Goncharov.}
``Cluster ensembles, quantization and the dilogarithm.''
\emph{Annales scientifiques de l'\'Ecole Normale Sup\'erieure} 42, no. 6 (2009): 865-930.

\bibitem[FI13]{FI13}
{I.\,Frenkel, I.\,Ip.}
``Positive representations of split real quantum groups and future perspectives.'' 
\emph{International Mathematics Research Notices} (2013): rns288.

\bibitem[FK94]{FK94}
{L.\,Faddeev, R.\,Kashaev.}
``Quantum dilogarithm.'' 
\emph{Modern Physics Letters A} 9, no. 5 (1994): 427-434.

\bibitem[FKV01]{FKV01}
{L.\,Faddeev, R.\,Kashaev, A.\,Volkov.}
``Strongly Coupled Quantum Discrete Liouville Theory I: Algebraic Approach and Duality.'' \emph{Communications in Mathematical Physics} 219, no. 1 (2001): 199-219.

\bibitem[FM16]{FM16}
{V.\,Fock, A.\,Marshakov.}
``Loop groups, clusters, dimers and integrable systems.''
\emph{Geometry and Quantization of Moduli Spaces} (2016): 1-65.

\bibitem[FT15]{FT15}
{L.\,Takhtajan, L.\,Faddeev.}
``The spectral theory of a functional-difference operator in conformal field theory.''
\emph{Izvestiya: Mathematics} 79, no. 2 (2015): 388-410.

\bibitem[Giv97]{Giv97}	
{A.\,Givental}
``Stationary phase integrals, quantum Toda lattices, flag manifolds and the mirror conjecture''.
\emph{American Mathematical Society Translations: Series 2} 180 (1997): 103-115.

\bibitem[GK11]{GK11}
{A.\,Goncharov, R.\,Kenyon.}
``Dimers and cluster integrable systems.''
\emph{Annales scientifiques de l'\'Ecole Normale Sup\'erieure} 46, no. 5 (2013): 747-813.

\bibitem[GKLO06]{GKLO06}
{A.\,Gerasimov, S.\,Kharchev, D.\,Lebedev, S.\,Oblezin. }
``On a Gauss-Givental representation for quantum Toda chain wave function.''
\emph{International Mathematics Research Notices} 2006 (2006).

\bibitem[GKLO08]{GKLO08}
{A.\,Gerasimov, D.\,Lebedev, S.\,Oblezin.}
``Baxter operator and Archimedean Hecke algebra.''
\emph{Communications in mathematical physics} 284, no. 3 (2008): 867-896.

\bibitem[GKLO14]{GKLO14}
{A.\,Gerasimov, D.\,Lebedev, S.\,Oblezin.}
``Baxter operator formalism for Macdonald polynomials.''
\emph{Letters in Mathematical Physics} 104, no. 2 (2014): 115-139.

\bibitem[Gon05]{Gon05}
{A.\,Goncharov.}
``Pentagon relation for the quantum dilogarithm and quantized $\mathcal M_{0,5}^{\mathrm{cyc}}$.''
\emph{Geometry and dynamics of groups and spaces} (2007): 415-428.

\bibitem[GSV13]{GSV13}
{M.\,Gekhtman, M.\,Shapiro, A.\,Vainshtein.}
``Generalized B\"acklund--Darboux transformations for Coxeter--Toda flows from a cluster algebra perspective.'' \emph{Acta mathematica} 206, no. 2 (2011): 245-310.

\bibitem[Gus94]{Gus94}
{R.\,Gustafson.}
``Some $q$-beta and Mellin--Barnes integrals on compact Lie groups and Lie algebras.''
\emph{Transactions of the American Mathematical Society} 341, no. 1 (1994): 69-119.

\bibitem[GT18]{GT18}
{R.\,Gonin, A.\,Tsymbaliuk.} ``On Sevostyanov's construction of quantum difference Toda lattices for classical groups.'' arXiv preprint arXiv:1804.01063 (2018).

\bibitem[HKKR00]{HKKR00}
{T.\,Hoffmann, J.\,Kellendonk, N.\,Kutz, N.\,Reshetikhin.}
``Factorization Dynamics and Coxeter--Toda Lattices.'' 
\emph{Communications in Mathematical Physics} 212, no. 2 (2000): 297-321.

\bibitem[Ip12a]{Ip12a}
{I.\,Ip.}
``Positive representations of split real simply-laced quantum groups.''
{\emph arXiv:1203.2018} (2012).

\bibitem[Ip15]{Ip15}
{I.\,Ip.}
``Positive representations of non-simply-laced split real quantum groups.''
\emph{Journal of Algebra} 425 (2015): 245-276.

\bibitem[Kas01]{Kas01}
{R.\,Kashaev.}
``The Quantum Dilogarithm and Dehn Twists in Quantum Teichm\"uller Theory.''
\emph{Integrable Structures of Exactly Solvable Two-Dimensional Models of Quantum Field Theory} (2001): 211-221.

\bibitem[KLS02]{KLS02}
{S.\,Kharchev, D.\,Lebedev, M.\,Semenov-Tian-Shansky.}
``Unitary Representations of $U_q(SL(2,\R))$, the Modular Double and the Multiparticle $q$-Deformed Toda Chain.''
\emph{Communications in mathematical physics} 225, no. 3 (2002): 573-609.

\bibitem[Kos79]{Kos79}
{B.\,Kostant. }
``Quantization and representation theory.''
\emph{Representation theory of Lie groups} (1979): 287-316.

\bibitem[Koz15]{Koz15}
{K.\,Kozlowski.}
``Unitarity of the SoV transform for the Toda chain.''
\emph{Communications in Mathematical Physics} 334, no. 1 (2015): 223-273.

\bibitem[Lam13]{Lam13}
{T.\,Lam.}
``Whittaker functions, geometric crystals, and quantum Schubert calculus.''
\emph{arXiv:1308.5451} (2013).

\bibitem[Rai09]{Rai09}
{E.\,Rains}
``Limits of elliptic hypergeometric integrals.''
\emph{The Ramanujan Journal} 18, no. 3 (2009): 257-306.

\bibitem[Rai10]{Rai10}
{E.\,Rains}
``Transformations of elliptic hypergeometric integrals.''
\emph{Annals of Mathematics} (2010): 169-243.

\bibitem[Sev99]{Sev99}
{A.\,Sevostyanov.}
``Quantum deformation of Whittaker modules and the Toda lattice.''
\emph{Duke Mathematical Journal} 105, no. 2 (2000): 211-238.

\bibitem[Sem94]{Sem94}
{M.\,Semenov-Tian-Shansky.}
``Quantization of Open Toda Lattices.''
\emph{Dynamical Systems VII} (1994): 226-259.

\bibitem[Sil07]{Sil07}
{A.\,Silantyev.}
``Transition function for the Toda chain.''
\emph{Theoretical and Mathematical Physics} 150, no. 3 (2007): 315-331.

\bibitem[SS17]{SS17}
{G.\,Schrader, A.\,Shapiro.}
``Continuous tensor categories from quantum groups I: algebraic aspects.''
\emph{arXiv:1708.08107} (2017).

\bibitem[Sta01]{Sta01}
{E.\,Stade.}
``Mellin transforms of $GL(n,\R)$ Whittaker functions.''
\emph{American journal of mathematics} 123, no. 1 (2001): 121-161.

\bibitem[Sta02]{Sta02}
{E.\,Stade.}
``Archimedean $L$-factors on $GL(n) \times GL (n)$ and generalized Barnes integrals.''
\emph{Israel Journal of Mathematics} 127, no. 1 (2002): 201-219.

\bibitem[Vol05]{Vol05}
{A.\,Volkov.}
``Noncommutative hypergeometry.''
\emph{Communications in mathematical physics} 258, no. 2 (2005): 257-273.

\end{thebibliography}

\end{document}